\documentclass[onecolumn]{article}

\usepackage[utf8x]{inputenc} 
\usepackage{amsmath}
\usepackage{mathrsfs}
\usepackage{amsfonts}
\usepackage{amssymb}
\usepackage{graphicx}
\usepackage{titletoc}
\usepackage{lmodern}
\usepackage[T1]{fontenc}
\usepackage{textcomp}
\usepackage{color}
\usepackage{parskip}
\usepackage{etoolbox}
\usepackage{amsthm}
\usepackage{hyperref}

\newtheorem{proposition}{Proposition}
\newtheorem{assumption}{Assumption}

\newtheorem{definition}{Definition}

\newtheorem{remark}{Remark}

\newcommand{\MAP}[3]{{#1}:\mathbb{R}^{#2}\to\mathbb{R}^{#3}}      
\newcommand{\bkd}[1]{\mbox{block}\left\lbrace{#1}\right\rbrace}

\AtBeginEnvironment{array}{\setlength{\arraycolsep}{2pt}}
\AtBeginEnvironment{eqnarray}{\setlength{\arraycolsep}{2pt}}
\DeclareMathOperator{\sech}{sech}

\def\rea{\mathbb{R}}

\def\diag{\mbox{diag}}

\def\rank{\mbox{rank}}

\begin{document}

\title{Stabilization of physical systems via saturated controllers with only partial state measurements}
\author{Pablo Borja\thanks{Department of Cognitive Robotics. Faculty of Mechanical, Maritime and Materials Engineering (3mE). Delft University of Technology (TU Delft). Delft, The Netherlands.}, Carmen Chan-Zheng\thanks{Jan C. Willems Center for Systems and Control. Engineering and Technology Institute Groningen. Faculty of Science and Engineering. University of Groningen. Groningen, The Netherlands.}, and Jacquelien M.A. Scherpen$^{\dag}$\\ {\small \href{mailto:l.p.borjarosales@tudelft.nl}{\tt l.p.borjarosales@tudelft.nl}, \href{mailto:c.chan.zheng@rug.nl}{\tt c.chan.zheng@rug.nl}, \href{mailto:j.m.a.scherpen@rug.nl}{\tt j.m.a.scherpen@rug.nl}}}

\maketitle


\begin{abstract}
This paper provides a constructive passivity-based control approach to solve the set-point regulation problem for input-affine continuous nonlinear systems while considering saturation in the inputs. As customarily in passivity-based control, the methodology consists of two steps: energy shaping and damping injection. In terms of applicability, the proposed controllers have two advantages concerning other passivity-based control techniques: (i) the energy shaping is carried out without solving partial differential equations, and (ii) the damping injection is performed without measuring the passive output. The proposed methodology is suitable to control a broad range of physical systems, e.g., mechanical, electrical, and electro-mechanical systems. We illustrate the applicability of the technique by designing controllers for systems in different physical domains, where we validate the analytical results via simulations and experiments.
\end{abstract}

\textbf{Keywords.} Passivity-based control, port-Hamiltonian systems, Brayton-Moser equations, dynamic extension, damping injection.

\section{Introduction}\label{sec:1}

 The behavior of a physical system is ruled by its energy, the interconnection pattern among its elements, its dissipation, and the interaction with its environment. These components are the main ingredients of passivity-based control (PBC). Hence, this control approach arises as a natural choice to control a wide variety of physical systems while taking into account conservation laws and other physical properties of the system under study see, for instance, \cite{ORTbook,ORTcsm,GEObook,VANJEL}. \\[0.2cm]
Due to its versatility, PBC has proven to be a powerful control approach to solve different problems, such as set-point regulation or trajectory tracking \cite{ORTbook, ORTcsm, VAN}. However, the implementation of these controllers may be hampered by physical limitations such as the operation ranges of the actuators or unavailable state measurements due to the lack of sensors. To address these issues, we propose a PBC approach suitable to stabilize a class of passive systems, where the controller is saturated and does not require full state measurements. These properties can be helpful to protect the actuators of the system, to avoid undesired oscillations, or to deal with the lack of sensors to measure specific elements of the state.\\[0.2cm] 
The injection of damping into the closed-loop is essential to guarantee that the system converges to the desired configuration. However, the signals involved in this process are not always measurable, e.g., velocities in mechanical systems. In this regard, observers offer a solution to this problem; we refer the reader to \cite{venkatraman2010full} for the port-Hamiltonian (pH) approach and \cite{venkatraman2010speed} for a class of mechanical systems. Nevertheless, the implementation of observers in nonlinear systems can hinder the stability analysis of the closed-loop system. In this work, we avoid the use of observers by proposing a new state vector, the dynamics of which are designed to inject damping into the closed-loop system using only measurable signals. For mechanical systems, a similar method to inject damping while avoiding velocity measurements is adopted in \cite{KELLY, LOROBS} for the Euler-Lagrange (EL) approach and in \cite{dirksz2013tracking, WESBOR} for the pH framework. In this paper, we generalize the results reported in \cite{WESBOR} to passive systems in different physical domains, not necessarily in the pH approach. Some differences between the methodology proposed in this paper and the results reported in \cite{KELLY, LOROBS, dirksz2013tracking} are: 

\begin{itemize}
 \item [(i)] The resulting controllers are saturated.
 \item [(ii)] The use of the open-loop dissipative terms to improve the transient response of the closed-loop system. In particular, in mechanical systems, we exploit the natural damping to modify the damping of the closed-loop system without measuring velocities.
 \item [(iii)] The methodology encompasses, in addition to the EL and pH approaches, other representation of passive systems, such as systems presented by the Brayton-Moser (BM) equations.
\end{itemize}

%
Some examples of PBC techniques that deal with the saturation problem for mechanical systems are \cite{colbaugh1997global, alvarez2003semiglobal, escobar1999output, JIANG, LORSAT, WESBOR}. Our approach differs from the mentioned references in the following aspects: 
\begin{itemize}
 \item [(i)] We propose a PBC approach that is suitable to stabilize a broad class of passive systems. This contrasts with \cite{escobar1999output, JIANG, WESBOR}, where the controllers are designed for specific systems.
 \item [(ii)] The proposed controllers inject damping without measuring the passive output. In particular, for mechanical systems, this implies that the control law does not require velocity measurements, which differs from the results reported in \cite{colbaugh1997global}.
 \item [(iii)] In contrast to \cite{LORSAT}, we consider underactuated systems.
 \item [(iv)] The methodology proposed in this paper does not require any change of coordinates during the control design.
\end{itemize}

%
The main contributions of this paper are summarized below:
\begin{itemize}
 \item [\textbf{C1}] We present a generalized framework for controlling passive systems, i.e., we consider input-affine nonlinear passive systems. This class of systems encompasses, but is not limited to, some popular modeling approaches, such as the pH framework or the EL formalism. Hence, we provide a method to stabilize nonlinear systems in different physical domains and whose models are not restricted to a particular modeling approach.
 \item [\textbf{C2}] We propose a method that considers input saturation without jeopardizing the stability of the closed-loop system nor increasing the stability analysis complexity. Consequently, the class of systems that can be stabilized is not reduced by considering saturated inputs. 
  \item [\textbf{C3}] We exploit the natural dissipation of the system to improve the performance of the closed-loop system.
 \item [\textbf{C4}] We provide the analysis of particular cases of interest, such as mechanical systems and electrical circuits, where the controllers are designed without solving partial differential equations (PDEs).
\end{itemize}
The remainder of this paper is organized as follows: we provide some preliminaries and the problem setting in Section \ref{sec:2}. Then, Sections \ref{sec:3} and \ref{sec:4}
are devoted to the control design,  where we establish the main results of this work. In Section \ref{sec:5}, we study some particular cases of interest. While, in Section \ref{sec:6},
some examples are provided to illustrate the applicability of the methodology proposed in this work. We wrap-up this paper with some concluding remarks and future work in Section \ref{sec:7}.\\[0.15cm]
\textbf{Caveat:} to ease the readability and simplify the notation in the proofs contained in this paper, when clear from the context, we omit the arguments of the functions.\\[0.15cm]
\textbf{Notation}: we denote the $n\times n$ identity matrix as $I_n$. The symbol $\mathbf{0}$ denotes a vector or matrix of appropriate dimensions whose entries are zeros. The symbols $\mbox{diag}\{ \}$ and $\mbox{block}\{ \}$ are used to denote diagonal and block diagonal matrices, respectively. Consider a vector $x\in\rea^{n}$, a smooth function $f:\mathbb{R}^n\to \mathbb{R}$, and the mappings $F:\mathbb{R}^n\to\mathbb{R}^m$, $G:\mathbb{R}^n\to\mathbb{R}^{n\times m}$. We define the differential operator $\nabla_x f:=\dfrac{\partial f}{\partial x}$ and $\nabla^2_x f:=\dfrac{\partial^2 f}{\partial x^2}$. The $ij-$th element of the $n\times m$ Jacobian matrix of $F(x)$ is given by $(\nabla_x F)_{ij}:=\dfrac{\partial F_j}{\partial x_i}$. We omit the subindex in $\nabla$ when it is clear from the context. Given the distinguished element $x_* \in \rea^n$, we define the constant vectors $F_*:=F(x_*)\in \rea^{m}$, $\left( \nabla f \right)_{*}:=\left. \nabla_{x}f(x) \right|_{x=x_{*}}$, and the constant matrices $G_*:=G(x_*)\in \rea^{n\times m}$,  $\left( \nabla F \right)_{*}:=\left. \nabla_{x}F(x) \right|_{x=x_{*}}$. Consider a matrix $A\in\rea^{n\times n}$. We say that $A$ is \textit{positive semi-definite}, denoted as $A\geq0$, if $A=A^{\top}$ and $x^{\top}Ax\geq 0$ for all $x\in\rea^{n}$, and \textit{positive definite}, denoted as $A>0$, if its symmetric and $x^{\top}Ax> 0$ for all $x\in\rea^{n}-\{\mathbf{0} \}$. We denote the Euclidean norm as $\lVert x \rVert$, i.e., $\lVert x \rVert = \sqrt{x^{\top}x}$, and the weighted Euclidean norm as $\lVert x \rVert_{A}:=\sqrt{x^{\top}Ax}$, where $A$ is positive (semi-)definite. The symbol $e_{i}$ denotes the $i$-th element of the canonical basis of $\rea^{a}$, where the context determines $a$, i.e., $e_{i}$ is a column vector such that its $i$-th element is one and the rest are zero.
\section{Preliminaries and problem setting}\label{sec:2}
Consider the input-affine nonlinear system
\begin{equation}
\begin{array}{rcl}
 \dot{x}&=&f(x)+g(x)u \\
\end{array}
  \label{afsys}
\end{equation} 
where $x\in\mathcal{X}\subseteq\rea^{n}$ is the state vector, $u\in \rea^{m}$ is the input, with $m\leq n$, $f:\mathcal{X}\to \rea$ denotes the so-called \textit{drift} vector field, 
$g:\mathcal{X}\to\rea^{n\times m}$ is the input matrix, which satisfies that $\rank\{ g(x)\}=m$.\\[0.2cm]
A broad class of physical systems, in different domains, can be described by the dynamics given in \eqref{afsys}. 
In this work, we are interested in the design of controllers that solve the set-point regulation problem for a class of nonlinear systems which admit the representation \eqref{afsys}. 
Therefore, our aim consists of ensuring that the closed-loop system has an asymptotically stable equilibrium at the desired point. 
Accordingly, the first step to formally formulate the control problem is to identify which points can be assigned as equilibria of the closed-loop system. 
Towards this end, we define the set that characterizes the assignable equilibria for the system \eqref{afsys}, which is given by
\begin{equation*}
 \mathcal{E}:=\left\lbrace x\in \mathcal{X} | g^{\perp}(x)f(x)=\mathbf{0} \right\rbrace,
\end{equation*} 
where $g^{\perp}:\mathcal{X}\to \rea^{(n-m)\times n}$ is the left annihilator of $g(x)$, i.e., $g^{\perp}(x)g(x)=\mathbf{0}$.\\[0.15cm]
There exist several nonlinear control design techniques that solve the set-point regulation problem. 
However, the implementation of these techniques is sometimes hampered by physical limitations, which are not considered by the controller. 
Two common problems that hinder the practical implementation of such controllers are: 
\begin{itemize}
 \item The lack of sensors to measure some relevant signals, for instance, the passive output, which is often necessary to inject damping into the closed-loop system.
 \item The necessity of saturated control signals to ensure the safety of the equipment or to avoid undesired transient behaviors due to the limited working range of the actuators.
\end{itemize}
The objective of this work is to propose controllers that regulate physical systems at the corresponding desired point while overcoming the issues mentioned above.
Below we set the control problem.\\[0.2cm] 
\textbf{Problem formulation.} Given the system \eqref{afsys}, propose a systematic control design approach such that:
\begin{itemize}
 \item The closed-loop system has a locally asymptotically stable equilibrium at the desired equilibrium $x_{*}\in\mathcal E $.
 \item The elements of the control law $u$ are saturated, i.e.,
 $u_{i}\in[\mathcal{U}_{min},\mathcal{U}_{max}]$, where the limits of the interval are bounded and can be chosen.
 \item The controller injects damping into the closed-loop system without measuring the passive output.
\end{itemize}

\section{Control design}\label{sec:3}

From a theoretical point of view, developing a general control design approach to stabilize systems that admit the representation given in \eqref{afsys} is, at best, a challenging task. Nonetheless, when dealing with physical systems, we can take advantage of some of their inherent properties. In particular, the passivity property exhibited by most of these systems can be exploited for control purposes. A thorough exposition of passive and cyclo-passive systems can be found in \cite{hill1976, VAN}. Here, for the sake of completeness, we provide the following definition of passive and cyclo-passive systems.
\begin{definition}\em
 The system \eqref{afsys} is said to be \textit{passive} if there exists a function $S:\rea^{n}\to \rea_+$, called the \textit{storage function}, and a signal $y\in\rea^{m}$, refer to as the \textit{passive output}, such that for all initial conditions $x(0)=x_{0}\in\rea^{n}$ the following inequality holds
\begin{equation}
 S(x(t))\leq S(x_{0})+\int_{{0}}^{t}u^{\top}(s)y(s)ds. \label{dissipativity}
\end{equation} 
Moreover, \eqref{afsys} is said to be \textit{cyclo-dissipative} if the storage function is not necessarily nonnegative, i.e., $S:\rea^{n}\to \rea$.
\end{definition}

\textit{Energy} and \textit{dissipation} play an essential role in the behavior of (cyclo-)passive systems. Consequently, \textit{energy-based} controllers, such as the ones derived from PBC techniques,
represent a suitable choice to control physical passive systems while preserving some physical intuition during the control design process. In this section, we develop a PBC approach that complies with the
requirements established in Section \ref{sec:2}. Hence, the first step consists in proving that the system under study is (cyclo-)passive. To this end, we characterize the input-affine nonlinear systems that are cyclo-passive via the following assumtpion.
\begin{assumption}\label{ass1}
 Given the system \eqref{afsys}, there exists $S:\mathcal{X}\to \rea$ such that
 \begin{equation}
   \left[\nabla S(x)\right]^{\top}f(x)  = - \lVert\ell(x)\rVert^{2}, \label{pascond}
 \end{equation} 
 where $\ell: \mathcal{X} \to \rea^{r}$ for some positive integer $r$.
\end{assumption}
Assumption \ref{ass1} is closely related to Hill-Moylan's theorem \cite{hill1976}, which provides necessary and sufficient conditions to determine whether \eqref{afsys} is cyclo-passive or not. However, at this point, the output of the plant has not been defined yet. Thus, it is not possible to establish the (cyclo)-passivity property of \eqref{afsys}. The proposition below establishes that Assumption \ref{ass1} guarantees that \eqref{afsys} is cyclo-passive and provides the structure of the passive output corresponding to the storage function $S(x)$.
\begin{proposition}\em
 Consider the system \eqref{afsys} satisfying Assumption \ref{ass1}, and some mappings $w:\mathcal{X} \to \rea^{r\times m}$, $D:\mathcal{X} \to \rea^{m\times m}$. Then,
 \begin{equation*}
  \dot{S}\leq u^{\top}y
 \end{equation*} 
 with 
 \begin{equation}
  y= g^{\top}(x)\nabla S(x)+2w^{\top}(x)\ell(x)+[w^{\top}(x)w(x)+D(x)]u, \label{ypas}
 \end{equation} 
 where $D(x)=-D^{\top}(x)$.
\end{proposition}
\begin{proof}
 Compute the derivative of $S(x)$ along the trajectories of \eqref{afsys}, that is,
 \begin{equation}
  \begin{array}{rcl}
  \dot{S} &=& \left(\nabla S\right)^{\top}(f+gu) \\[0.1cm]
   &=& -\lVert\ell\rVert^{2}+ y^{\top}u - 2\ell^{\top} wu - u^{\top}w^{\top}wu\\[0.1cm]
   &=& - \lVert \ell + wu \rVert^{2}+y^{\top}u\leq y^{\top}u,
  \end{array}\label{dS}
 \end{equation} 
 where we used \eqref{ypas} and $u^{\top}Du=0$.
\end{proof}
Customarily, PBC techniques consist of two steps: first, the so-called \textit{energy-shaping} where the new energy---storage---function is modified to have a minimum at the desired equilibrium. Second, the damping injection into the closed-loop system ensures that the trajectories converge to the desired point. We present the following assumption to characterize the class of systems for which the controllers devised in this section can assign the desired equilibrium to the closed-loop system and render it stable.\\
\begin{assumption}\label{ass2}
Consider \eqref{afsys} satisfying Assumption \ref{ass1}, and the desired equilibrium $x_{*}\in\mathcal{E}$.
There exists $\gamma:\mathcal{X}\to \rea^{m}$ such that
 \begin{equation}
  \begin{array}{l}
  \dot{\gamma}  =  y\\[0.2cm]
   \left( \nabla S \right)_{*}+\left( \nabla \gamma \right)_{*}\kappa=\mathbf{0}\\[0.2cm]
  \left( \nabla^{2} S \right)_{*}+\displaystyle\sum_{i=1}^{m}\left( \nabla \gamma_{i} \right)_{*}\left( \nabla \gamma_{i} \right)^{\top}_{*}\alpha_{i}\beta_{i}+
  \left( \nabla^{2}\left(\gamma^{\top}\kappa\right) \right)_{*}>0. \label{hesscond}
  \end{array}
 \end{equation} 
where $\alpha_{i},\beta_{i}$ are positive constants and
 \begin{equation}
  \kappa:=(g_{*}^{\top}g_{*})^{-1}g_{*}^{\top}f_{*}. \label{kappa}
 \end{equation}
\end{assumption}
Proposition \ref{pro1} provides a saturated controller that addresses the stabilization problem for systems that satisfy Assumptions \ref{ass1} and \ref{ass2}.
\begin{proposition}\label{pro1}\em
Suppose that the system \eqref{afsys} and the desired equilibrium $x_{*}\in\mathcal{E}$ satisfy Assumptions \ref{ass1} and \ref{ass2}. Consider the control law 
 \begin{equation}
  u = -\nabla_{\gamma}\Phi(\gamma(x))-\kappa - \displaystyle\sum_{i=1}^{m}e_{i}k_{p_{i}}\tanh(y_{i}), \label{u}
 \end{equation} 
 where $k_{p_{i}}>0$, for $i=1,\dots,m$, and
 \begin{equation}
 \Phi(\gamma(x)):= \displaystyle\sum_{i=1}^{m}\frac{\alpha_{i}}{\beta_{i}}\ln\left(\cosh(\beta_{i}\left(\gamma_{i}(x)-\gamma_{i*}\right))\right). \label{Phi}
\end{equation} 
 Then:
 \begin{itemize}
  \item [(i)] The control signals satisfy $$u_{i}\in[-\kappa_{i}-k_{p_{i}}-\alpha_{i},-\kappa_{i}+k_{p_{i}} + \alpha_{i}].$$
  \item [(ii)] The closed-loop system has a \textit{locally stable} equilibrium point at $x_{*}$
  with Lyapunov function
  \begin{equation}
   S_{d}(x):= S(x) + \Phi(\gamma(x))+\kappa^{\top}\gamma(x). \label{Sd}
  \end{equation} 
  \item [(iii)] The equilibrium is \textit{locally asymptotically stable} if, on a domain $\Omega\subseteq\mathcal{X}$ containing $x_{*}$,
  \begin{equation}
   \begin{array}{rcl}
    \left.\begin{array}{rcl}
     \ell &=& w\left(\nabla_{\gamma}\Phi(\gamma(x))+\kappa\right) \\
      y &=& \mathbf{0}
    \end{array}\right\rbrace \implies x=x_{*}
   \end{array} \label{detect}
  \end{equation} 
 \end{itemize}
\end{proposition}
\begin{proof}
 To prove (i) note that
 \begin{equation}
  \nabla_{\gamma}\Phi(\gamma(x)) = \displaystyle\sum_{i=1}^{m}e_{i}\alpha_{i}\tanh(\beta_{i}\left( \gamma_{i}(x)-\gamma_{i*} \right)). \label{nPhi}
 \end{equation} 
 Hence, the control law takes the form
 \begin{equation*}
  u_{i}= -\kappa_{i}-\alpha_{i}\tanh(\beta_{i}\left( \gamma_{i}(x)-\gamma_{i*} \right))-k_{p_{i}}\tanh(y_{i}).
 \end{equation*} 
 Since the function $\tanh(\cdot)$ is saturated, we get that
   \begin{equation}
   -\kappa_{i} - k_{p_{i}} - \alpha_{i} \leq u_{i} \leq  -\kappa_{i} + k_{p_{i}} + \alpha_{i}. \label{interval}
  \end{equation}
 To prove (ii) we compute
 \begin{eqnarray}
 \nonumber \dot{S}_{d} & = & -\lVert \ell + wu \rVert^{2} + y^{\top}u + \dot{\gamma}^{\top}\nabla_{\gamma}\Phi + \dot{\gamma}^{\top}\kappa \\[0.1cm]
   & = & -\lVert \ell + wu \rVert^{2} + y^{\top}\left( u + \nabla_{\gamma}\Phi + \kappa \right), \label{dSd1}
 \end{eqnarray} 
 where we used \eqref{dS} and \eqref{nPhi}. Substituting \eqref{u} in \eqref{dSd1}, yield
 \begin{equation}
  \dot{S}_{d}  \leq -y^{\top}\displaystyle\sum_{i=1}^{m}e_{i}k_{p_{i}}\tanh(y_{i}) = -\displaystyle\sum_{i=1}^{m}k_{p_{i}}y_{i}\tanh(y_{i})\leq 0 \label{dSd2}.
 \end{equation} 
Furthermore, from Assumption \ref{ass2}, we have that
\begin{equation}
 \left( \nabla S_{d}\right)_{*} = \left( \nabla S \right)_{*} + \left( \nabla\gamma \right)_{*}\kappa=\mathbf{0}, \label{eq4}
\end{equation} 
and
\begin{equation}
\begin{array}{rcl}
  \left(\nabla^{2} S_{d}\right)_{*} 
   &=&  \displaystyle\sum_{i=1}^{m}\left( \nabla \gamma_{i} \right)_{*}\left( \nabla \gamma_{i} \right)^{\top}_{*}\alpha_{i}\beta_{i}+\left( \nabla^{2}\left(\gamma^{\top}\kappa\right) \right)_{*}+\left(\nabla^{2} S \right)_{*}>0. 
  \end{array}
  \label{HessSd}
\end{equation} 
Hence, $x_{*}$ is an isolated minimum of $S_{d}(x)$. Moreover, \eqref{dSd2} implies that $S_{d}(x)$ is non-increasing. Thus, $S_{d}(x)$ is positive definite with respect to $x_{*}$. Accordingly, $S_{d}(x)$ is a Lyapunov function and $x_{*}$ is stable.\\[0.2cm] 
To prove (iii), note that \eqref{dSd1} and \eqref{dSd2} yield
\begin{equation}
 \dot{S}_{d}  = 0 \Longleftrightarrow \left\lbrace \begin{array}{rcl}
\ell+wu&=&\mathbf{0}\\[0.1cm]
y & = & \mathbf{0}.
\end{array}\right. \label{u1detcond}
\end{equation} 
However, 
\begin{equation*}
 y  =  \mathbf{0} \implies u=-\left(\nabla_{\gamma}\Phi(\gamma(x))+\kappa\right).
\end{equation*} 
Therefore, 
\begin{equation*}
 \dot{S}_{d} =0 \iff \left\lbrace \begin{array}{rcl}
\ell&=&w\left(\nabla_{\gamma}\Phi(\gamma(x))+\kappa\right)\\[0.1cm]
y & = & \mathbf{0}.
\end{array}\right.
\end{equation*} 
Furthermore, 
\begin{equation*}
  x=x_{*}\implies \dot{S}_{d} = 0.
\end{equation*} 
Hence, \eqref{detect} implies that, on the domain $\Omega$,
\begin{equation*}
 \dot{S}_{d} = 0 \iff x=x_{*}.
\end{equation*} 
Thus, the asymptotic stability of $x_{*}$ follows from Barbashin-Krasovskii's theorem, see Corollary 4.1 in \cite{KHA}.
\end{proof}
Note that the saturation limits of the control law \eqref{u} can be adjusted by modifying the control parameters $\alpha_{i}$ and $k_{p_{i}}$. Furthermore, we point out that the natural dissipation plays an important role in the stabilization of the system. In particular, we make the following remarks.
\begin{remark}
 If $\ell_{*}\neq \mathbf{0}$, then the desired equilibrium can be assigned only by shaping the energy using a $\gamma$ derived from a passive output with \textit{relative degree} zero.
 A proof of this fact can be found in \cite{BORSCHERCDC18}. This phenomenon is called the \textit{dissipation obstacle}. We refer the reader to \cite{MENetal} for further details on this topic.
\end{remark}
\begin{remark}
 If 
 \begin{equation}
  \ell  = w\left(\nabla_{\gamma}\Phi(\gamma(x))+\kappa\right) \implies x=x_{*}, \label{d1}
 \end{equation} 
 then it is not necessary to inject damping into the closed-loop system to ensure the asymptotic stability of the desired equilibrium. Moreover, 
 \begin{equation*}
  u = -\kappa - \nabla_{\gamma}\Phi(\gamma(x))
 \end{equation*}
 solves the regulation problem. On the other hand, if $\ell = \mathbf{0}$, then \eqref{detect} reduces to
 \begin{equation}
  y=\mathbf{0} \implies x=x_{*}, \label{d2}
 \end{equation} 
 In particular, when $x_{*}=\mathbf{0}$, \eqref{d2} is referred to as \textit{zero-state observability}, see \cite{ORTbook, KHA}. We stress the fact that \eqref{d2} is more conservative than \eqref{detect}. Indeed, proving that \eqref{d2} holds, would suffice to claim asymptotic stability of the equilibrium point in Proposition \ref{pro1}. 
\end{remark}
The control law \eqref{u} addresses the regulation problem and ensures that the control signals are saturated, where the damping is injected through the passive output. However, the measurement of this signal is not always available, e.g., in mechanical systems without velocity sensors. To overcome this issue, we propose a modified control law such that the damping injection does not require the measurement of $y$. To this end, we introduce the controller state $x_{c}\in\rea^{m}$, and we define the following mappings
\begin{equation}
 \begin{array}{rcl}
z_{c}(\gamma(x),x_{c})&:=& \gamma(x) - \gamma_{*} + x_{c} \\ \Phi_{c}(z_{c}(\gamma(x)+x_{c}))&:=& \displaystyle\sum_{i=1}^{m}\frac{\alpha_{c_{i}}}{\beta_{c_{i}}}\ln\left(\cosh(\beta_{c_{i}}z_{c_{i}}(\gamma_{i}(x),x_{c_{i}}))\right), \end{array} \label{Phic}
\end{equation} 
where $\alpha_{c_{i}},\beta_{c_{i}}$ are positive constants. Without loss of generality, we consider $x_{c_{*}}=\mathbf{0}$. Moreover, to simplify the notation, we omit the argument of $z_{c}$.\\[0.2cm]
The following proposition provides a saturated control law that shapes the energy of the closed-loop system and injects damping without measuring $y$.
\begin{proposition}\label{pro2}\em
 Suppose that the system \eqref{afsys} and the desired equilibrium $x_{*}\in\mathcal{E}$ satisfy Assumptions \ref{ass1} and \ref{ass2}. Fix $\alpha_{c}=\alpha$ and $\beta_{c}=\beta$. Consider the positive definite matrices $R_{c}, K_{c}\in\rea^{m\times m}$, the dynamics
 \begin{equation}
  \dot{x}_{c} = -R_{c}\left[ \displaystyle\sum_{i=1}^{m}e_{i}\alpha_{c_{i}}\tanh(\beta_{c_{i}}z_{c_{i}}) +K_{c}x_{c}  \right],\label{dxc}
 \end{equation} 
 and the control law
 \begin{equation}
  u=-\kappa-\displaystyle\sum_{i=1}^{m}e_{i}\alpha_{c_{i}}\tanh(\beta_{c_{i}}z_{c_{i}}). \label{uc}
 \end{equation}
 Then:
 \begin{itemize}
  \item [(i)] The control signals satisfy $u_{i}\in[-\kappa_{i}-\alpha_{c_{i}},-\kappa_{i}+\alpha_{c_{i}}]$.
  \item [(ii)] There exists $K_{c}$ such that the closed-loop system has a locally asymptotically stable equilibrium point at $(x_{*},\mathbf{0})$  with Lyapunov function
  \begin{equation}
   S_{d_{c}}(x,x_{c}):= S(x) + \Phi_{c}(z_{c})+\kappa^{\top}\gamma(x)+\frac{1}{2}\lVert x_{c}\rVert^{2}_{K_{c}}, \label{Sdc}
  \end{equation}
  where $\Phi_{c}(z_{c})$ is defined as in \eqref{Phic}.
  \item [(iii)] The equilibrium is locally asymptotically stable if, on a domain $\Omega_{c}\subseteq\mathcal{X}\times\rea^{m}$ containing $(x_{*},\mathbf{0})$, the following condition holds 
  \begin{equation}
   \begin{array}{rcl}
    \left.\begin{array}{rcl}
     \ell &=& w\left( \kappa-K_{c}x_{c} \right) \\
      y &=& \mathbf{0}
    \end{array}\right\rbrace \implies \left\lbrace\begin{array}{rcl}
      x&=&x_{*} \\ x_{c}&=&\mathbf{0}                                      
                                           \end{array}\right.
   \end{array} \label{detectc}
  \end{equation} 
\end{itemize}
\end{proposition}
\begin{proof}
To prove (i) note that
\begin{equation*}
 u_{i} = - \kappa_{i} - \alpha_{c_{i}}\tanh(\beta_{c_{i}}z_{c_{i}}). 
\end{equation*} 
Thus
   \begin{equation}
   -\kappa_{i} - \alpha_{c_{i}} \leq u_{i} \leq  -\kappa_{i}  + \alpha_{c_{i}}. \label{intervaly}
  \end{equation} 
To prove (ii) note that
\begin{equation}
 \nabla_{\gamma}\Phi_{c} = \nabla_{x_{c}}\Phi_{c} = 
 \displaystyle\sum_{i=1}^{m}e_{i}\alpha_{c_{i}}\tanh(\beta_{c_{i}}z_{c_{i}}). \label{nPhic}
\end{equation} 
Hence,
\begin{eqnarray}
 \nonumber \dot{S}_{d_{c}}&=& -\lVert \ell + wu \rVert^{2} + y^{\top}\left( u+ \kappa \right) + \dot{z}_{c}^{\top}\nabla_{x_{c}}\Phi_{c}+\dot{x}_{c}^{\top}K_{c}x_{c}\\[0.1cm]
 &=&-\lVert \ell + wu \rVert^{2}  - \left\lVert \dot{x}_{c}\right\rVert^{2}_{R^{-1}_{c}}\leq 0 \label{dSdc}
\end{eqnarray} 
where we used \eqref{dS}, \eqref{dxc}, \eqref{uc}, and \eqref{nPhic}. Furthermore, since $z_{c_{*}}=\mathbf{0}$, we have
\begin{equation}
 \left( \nabla_{x} S_{d_{c}} \right)_{*}  =  \left( \nabla S \right)_{*} + \left( \nabla \gamma \right)_{*}\kappa \label{nSdc}.
\end{equation} 
Therefore, from Assumption \ref{ass2}, $\left( \nabla_{x} S_{d_{c}} \right)_{*} = \mathbf{0}$.
Moreover,
\begin{equation*}
  \nabla_{x_{c}} S_{d_{c}} = \nabla_{x_{c}}\Phi_{c} + K_{c}x_{c}. 
\end{equation*} 
Accordingly, $\left( \nabla_{x_{c}} S_{d_{c}}\right)_{*}=\mathbf{0}$.
Hence, $(x_{*}, \mathbf{0})$ is a critical point of $S_{d_{c}}(x,x_c)$. Furthermore, some simple computations show that
\begin{equation}
\begin{array}{rcl}
\left( \nabla^{2}S_{d_{c}} \right)_{*}&=&\begin{bmatrix}
                      \left( \nabla^{2}S \right)_{*}+\left( \nabla^{2}\left(\gamma^{\top}\kappa\right) \right)_{*} && \mathbf{0} \\ \mathbf{0} && K_{c}                 
                                       \end{bmatrix}+\begin{bmatrix}
                                                       \left( \nabla \gamma \right)_{*} \\ I_{m}
                                                     \end{bmatrix}\displaystyle\sum_{i=1}^{m}e_{i}e_{i}^{\top}\alpha_{c_{i}}\beta_{c_{i}}
\begin{bmatrix}
                                                      \left( \nabla \gamma \right)^{\top}_{*} & I_{m}
                                                     \end{bmatrix}.\end{array}
 \label{hessSdc}
\end{equation} 

Note that the block $(1,1)$ of $\left( \nabla^{2}S_{d_{c}} \right)_{*}$ can be expressed as $\left( \nabla^{2}S_{d} \right)_{*}$; see \eqref{HessSd}. Consequently, the blocks $(1,1)$ and $(2,2)$ of $\left( \nabla^{2}S_{d_{c}} \right)_{*}$ are positive definite. Moreover, $K_{c}$ is a control parameter whose only restriction is to be positive definite. Therefore, a Schur complement analysis shows that a $K_{c}$ large enough ensures that $\left( \nabla^{2}S_{d_{c}} \right)_{*}>0$.
Hence, an appropriate selection of $K_{c}$ ensures that $S_{d_{c}}(x,x_{c})$ has an isolated minimum at $(x_{*},\mathbf{0})$. This, together with \eqref{dSdc}, implies that $S_{d_{c}}(x,x_{c})$ is positive definite with respect to the equilibrium. Thus, $S_{d_{c}}(x,x_{c})$ is a Lyapunov function and the equilibrium is stable.\\[0.1cm]
To prove (iii) note that, from \eqref{nPhic},
\begin{equation}
 \begin{array}{rcl}
  \dot{S}_{d_{c}}  =  0 & \Longleftrightarrow & \left\lbrace\begin{array}{l}
\ell + wu = \mathbf{0}_{r} \\[0.1cm]
\dot{x}_{c} = \mathbf{0} \implies \nabla_{x_{c}}\Phi_{c}+K_{c}x_{c}= \mathbf{0}.
  \end{array}\right. 
 \end{array}\label{dSdceq}
\end{equation} 
From \eqref{dSdceq}, we have the following chain of implications
\begin{equation}
 \begin{array}{ccl}
&&\displaystyle\frac{d}{dt}\left( \nabla_{x_{c}}\Phi_{c}+K_{c}x_{c} \right)= \mathbf{0} \\[0.5cm]
&\implies& \displaystyle\sum_{i=1}^{m}e_{i}\dot{z}_{c_{i}}\alpha_{c_{i}}\left[\sech(\beta_{c_{i}}z_{c_{i}})\right]^{2}+K_{c}\dot{x}_{c}= \mathbf{0} \\[0.5cm]
&\implies& \displaystyle\sum_{i=1}^{m}e_{i}\dot{\gamma}_{i}\alpha_{c_{i}}\left[\sech(\beta_{c_{i}}z_{c_{i}})\right]^{2}= \mathbf{0}\\[0.5cm]
&\implies& y=\mathbf{0}.
 \end{array} \label{dampxc}
\end{equation} 
Moreover, \eqref{uc} takes the form
\begin{equation*}
 \begin{array}{rcl}
  u &=& -\kappa +K_{c}x_{c}.
 \end{array} 
\end{equation*} 
Hence, the asymptotic stability of the equilibrium can be proven using similar arguments as in the proof of Proposition \ref{pro1}.
\end{proof}
In the control law \eqref{uc}, the damping is injected via the controller state $x_{c}$. In particular, we propose the specific dynamics given in \eqref{dxc}. This can be interpreted as a dirty-derivative filter; see \cite{LOROBS} and \cite{dirksz2013tuning}. In contrast to the mentioned references, we extend this approach to a more general class of systems, i.e., input-affine nonlinear systems, while considering saturation in the inputs. We stress that an important consequence of \eqref{dampxc} is that the desired equilibrium is asymptotically stable if $y$ is observable, i.e., if \eqref{d2} is satisfied.\\[0.2cm]
The saturated controllers developed in this section address the regulation problem by shaping the energy of the system and injecting damping either through the passive output or the controller state $x_{c}$. In both cases, the damping injection is closely related to the output port. However, to improve the performance of the closed-loop system, it may be necessary to inject in coordinates that are not associated with $y$. In the following section, we provide an alternative to address this issue.
\section{On the role of the dissipation}\label{sec:4}
Dissipation is present in most physical systems. Nonetheless, the mathematical models that represent these systems commonly neglect the dissipation inherent to them. This section proposes a method to inject damping into the coordinates with natural dissipation without measuring them. This can be exploited to improve the convergence rate of the closed-loop system or to remove an undesired transient behavior, such as oscillations. Furthermore, this approach can be instrumental when the system under study exhibits poor damping propagation, resulting in a slow convergence rate.\\[0.2cm]
We characterize the systems for which this new damping injection is suitable through the following assumption. 
\begin{assumption}\label{ass3}
There exist $S:\mathcal{X}\to \rea$ and $\eta:\mathcal{X}\to\rea^{s}$, with $1\leq s \leq n-m$, such that the system \eqref{afsys} satisfies 
 \begin{equation*}
  \begin{array}{rcl}
  \left[\nabla\gamma(x)\right]^{\top}\nabla\eta(x) &=&\mathbf{0} \\[0.15cm]
   -\lVert \ell(x) + w(x)u \rVert^{2} &\leq& -\lVert \dot{\eta} \rVert_{\Lambda_{\ell}(x)}^{2} - \lVert y \rVert_{\Lambda_{c}(x)}^{2} 
  \end{array}
 \end{equation*} 
 where $\dot{\gamma}=y$ and the diagonal matrices $\Lambda_{\ell}:\mathcal{X}\to\rea^{s\times s}$, $\Lambda_{c}:\mathcal{X}\to\rea^{m\times m}$ are positive definite.
\end{assumption}
Assumption \ref{ass3} requires that some coordinates not associated with the output port are damped. Moreover, such damping must be decoupled from the rest of the coordinates. Some examples of this phenomenon are friction between surfaces, resistors in series with inductors, and resistors in parallel with capacitors. In order to use the mentioned damping in the control design, we introduce the virtual state $x_{\ell}\in\rea^{m}$, and the following mappings
\begin{eqnarray}
 \nonumber z_{\ell}(\eta(x),x_{\ell})&:=& \Upsilon\left[\eta(x)-\eta_{*} \right] + K_{\ell}x_{\ell}\\
 \Phi_{\ell}(\eta(x),x_{\ell})&:=&\displaystyle\sum_{i=1}^{m}\frac{\alpha_{\ell_{i}}}{\beta_{\ell_{i}}}\ln\left( \cosh(\beta_{\ell_{i}}z_{\ell_{i}}(\eta(x),x_{\ell})) \right) \label{Phil}
\end{eqnarray}
where the constant matrix $\Upsilon\in\rea^{m\times s}$ satisfies $\textnormal{rank}\{\Upsilon \} = \min\{m,s \}$, $K_{\ell}\in\rea^{m\times m}$ is a diagonal positive definite matrix, and $\alpha_{\ell_{i}},\beta_{\ell_{i}}$ are positive constant parameters. Without loss of generality,  we consider $x_{\ell_{*}}=\mathbf{0}$. To simplify the notation, we omit the argument from $z_{\ell}$.\\[0.2cm]
The following proposition provides a saturated control law that shapes the energy and modifies the damping of the coordinates that are naturally damped.
\begin{proposition}\label{pro4}\em
 Suppose that the system \eqref{afsys} and the desired equilibrium $x_{*}\in\mathcal{E}$ satisfy Assumptions \ref{ass1}--\ref{ass3}. Fix $\alpha_{c}=\alpha$, $\beta_{c}=\beta$, and consider the control law
\begin{equation}
  u=-\kappa-\nabla_{z_{\ell}}\Phi_{\ell}(z_{\ell})-\nabla_{\gamma}\Phi_{c}(z_{c}), \label{ulc}
 \end{equation}
 where $\Phi_{c}(z_{c})$ is defined in \eqref{Phic}, $\Phi_{\ell}(z_{\ell})$ is defined in \eqref{Phil}, the dynamics of $x_{c}$ are given by \eqref{dxc}, and
 \begin{equation}
  \dot{x}_{\ell} = -R_{\ell} \displaystyle\sum_{i=1}^{m}e_{i}\alpha_{\ell_{i}}\tanh(\beta_{\ell_{i}}z_{\ell_{i}}), \label{dxl}
 \end{equation}
 Then:
 \begin{itemize}
  \item [(i)] The control signals satisfy 
   $$u_{i}\in\left[-\kappa_{i}-\alpha_{\ell_{i}}-\alpha_{c_{i}},-\kappa_{i}+\alpha_{\ell_{i}}+\alpha_{c_{i}}\right].$$
  \item [(ii)] There exist $\ K_{c}, \ K_{\ell}$, and $R_{\ell}$ such that the closed-loop system has a locally asymptotically stable equilibrium point at $(x_{*},\mathbf{0},\mathbf{0})$,
  with Lyapunov function
    \begin{equation*}
    S_{d_{\ell}}(x,x_{\ell},x_{c}):= S_{d_{c}}(x,x_{c}) + \Phi_{\ell}(z_{\ell}),
  \end{equation*} 
  with $S_{d_{c}}(x,x_{c})$ defined in \eqref{Sdc}.
  \item [(iii)] For an appropriate selection of $K_{\ell}$ and $R_{\ell}$, the equilibrium is locally asymptotically stable if, on a domain $\Omega_{d}\subseteq\mathcal{X}\times\rea^{m}\times\rea^{m}$ containing $(x_{*},\mathbf{0},\mathbf{0})$, the following condition holds
  \begin{equation}
   \begin{array}{rcl}
    \left.\begin{array}{rcl}
     \dot{\eta} &=& \mathbf{0} \\
      y &=& \mathbf{0}\\
      \dot{x}_{\ell}&=&\mathbf{0} \\
      \dot{x}_{c}&=&\mathbf{0}
    \end{array}\right\rbrace \implies \left\lbrace\begin{array}{rcl}
      x&=&x_{*} \\ x_{\ell}&=&\mathbf{0} \\ x_{c}&=&\mathbf{0}.                                       
                                           \end{array}\right.
   \end{array} \label{detectlc}
  \end{equation} 
  \end{itemize}
\end{proposition}
\begin{proof}
 To prove (i) note that, from \eqref{nPhic} and \eqref{Phil}, the control law \eqref{ulc} can be rewritten as
 \begin{equation*}
  u = -\kappa -\displaystyle\sum_{i=1}^{m}e_{i}\left\lbrace \alpha_{c_{i}}\tanh\left( \beta_{c_{i}}z_{c_{i}} \right)+\alpha_{\ell_{i}}\tanh\left( \beta_{\ell_{i}}z_{\ell_{i}} \right)\right\rbrace.
 \end{equation*} 
 Therefore,
 \begin{equation*}
  -\kappa_{i}-\alpha_{c_{i}}-\alpha_{\ell_{i}} \leq u_{i}\leq -\kappa_{i}+\alpha_{c_{i}}+\alpha_{\ell_{i}}.
 \end{equation*}
 To prove (ii) note that \eqref{dxl} can be rewritten as 
 \begin{equation*}
  \dot{x}_{\ell} = -R_{\ell}\nabla_{z_{\ell}}\Phi_{\ell}.
 \end{equation*} 
 Hence, from \eqref{dSdc}, \eqref{ulc} and Assumption \ref{ass3}, we have that\footnote{Note that $R_{\ell}$ and $K_{\ell}$ are diagonal. Thus, their product commutes.}
 \begin{equation*}
  \begin{array}{rcl}
\dot{S}_{d_{\ell}}&\leq& -\lVert \dot{\eta} \rVert^{2}_{\Lambda_{\ell}} - \lVert y \rVert^{2}_{\Lambda_{c}} + \left(\dot{z_{\ell}}-y\right)^{\top}\nabla_{z_{\ell}}\Phi_{\ell} -\lVert \dot{x}_{c} \rVert^{2}_{R^{-1}_{c}} \\[0.1cm] &=&
-\lVert \dot{\eta} \rVert^{2}_{\Lambda_{\ell}} - \lVert y \rVert^{2}_{\Lambda_{c}} -\lVert \dot{x}_{c} \rVert^{2}_{R^{-1}_{c}} - \lVert \dot{x}_{\ell} \rVert^{2}_{K_{\ell}R^{-1}_{\ell}}+ \left( \Upsilon\dot{\eta} - y \right)^{\top}\nabla_{z_{\ell}}\Phi_{\ell}\\[0.1cm]
&=& -\begin{bmatrix}
        \dot{\eta}^{\top} & y^{\top} & \dot{x}_{\ell}^{\top}
       \end{bmatrix}\Theta 
\begin{bmatrix}
        \dot{\eta}^{\top} & y^{\top} & \dot{x}_{\ell}^{\top}
       \end{bmatrix}^{\top}-\lVert \dot{x}_{c} \rVert^{2}_{R^{-1}_{c}}.
  \end{array}
 \end{equation*}
 where 
  \begin{equation}
 \Theta:= \begin{bmatrix}
                    \Lambda_{\ell} && \mathbf{0} && \frac{1}{2}\Upsilon^{\top}R_{\ell}^{-1} \\[0.2cm] \mathbf{0} && \Lambda_{c} && -\frac{1}{2}R_{\ell}^{-1} \\[0.2cm]
                     \frac{1}{2}R_{\ell}^{-1}\Upsilon && -\frac{1}{2}R_{\ell}^{-1} &&  K_{\ell}R_{\ell}^{-1}.
                    \end{bmatrix}\label{Th}
 \end{equation} 
 Thus, $\dot{S}_{d_{\ell}}\leq 0$ if $\Theta$ is positive semi-definite. Moreover, via Schur complement, we get that $\Theta\geq0$ if and only if
 \begin{equation}
   \begin{bmatrix}
\Lambda_{\ell} & \mathbf{0}_{s\times m} \\ \mathbf{0} & \Lambda_{c}
\end{bmatrix}-\frac{1}{4}\begin{bmatrix}
               \Upsilon^{\top} \\ -I_{m}
              \end{bmatrix}K_{\ell}^{-1}R_{\ell}^{-1}\begin{bmatrix}
                                                                   \Upsilon & -I_{m}
                                                                  \end{bmatrix}
\geq 0.\label{schurcond1}
 \end{equation}
 Note that \eqref{schurcond1} holds for $K_{\ell}$ and $R_{\ell}$ large enough. Accordingly, an appropriate selection of these matrices ensures that $S_{d_{\ell}}(x,x_{\ell},x_{c})$ is non-increasing.
 Moreover, $\left(\nabla_{z_{\ell}}\Phi_{\ell}  \right)_{*}=\mathbf{0}$. This, together with $\left( \nabla S_{d_{c}} \right)_{*}=\mathbf{0}$---see the proof of Proposition \ref{pro2}---yields 
 \begin{equation*}
  \left( \nabla S_{d_{\ell}} \right)_{*}=\mathbf{0}.
 \end{equation*} 
 Furthermore, 
 \begin{equation*}
  \left( \nabla^{2} S_{d_{\ell}}\right)_{*} = \left( \nabla^{2} S_{d_{c}}\right)_{*}+\left( \nabla^{2} \Phi_{\ell}\right)_{*}.
 \end{equation*} 
Some simple computations show that $\left( \nabla^{2} \Phi_{\ell}\right)_{*}\geq0$ and, using the arguments of the proof of Proposition \ref{pro2}, a $K_{c}$ large enough guarantees that $\left( \nabla^{2} S_{d_{c}}\right)_{*}>0$. Therefore, there exists a $K_{c}$ such that $\left( x_{*},\mathbf{0},\mathbf{0} \right)$ is an isolated minimum of the closed-loop storage function. Therefore, $S_{d_{\ell}}(x,x_{\ell},x_{c})$ qualifies as a Lyapunov function and the closed-loop system has a stable equilibrium at $(x_{*}, \mathbf{0}, \mathbf{0})$.\\[0.2cm]
 To prove (iii), we consider $K_{\ell}$ and $R_{\ell}$ such that $\Theta>0$. Then, 
 \begin{equation*}
   \dot{S}_{d_{\ell}} = 0 \iff \left\lbrace\begin{array}{rcl}
     \dot{\eta} &=& \mathbf{0} \\
      y &=& \mathbf{0}\\
      \dot{x}_{\ell}&=&\mathbf{0} \\
      \dot{x}_{c}&=&\mathbf{0}
    \end{array}\right.
 \end{equation*} 
Thus, the asymptotic stability of the equilibrium is proven using the same arguments used in the proof of Proposition \ref{pro1}.
\end{proof}
\begin{remark}
 The stability properties of the equilibrium points in Propositions \ref{pro1}--\ref{pro4} are \textit{global} if the corresponding Lyapunov function is \textit{radially unbounded}, see \cite{KHA}.
\end{remark}
In general, to ensure the existence of $\gamma(x)$ in Assumption \ref{ass2}, it is necessary to solve a PDE. However, in some particular cases, $\gamma(x)$ can be found by satisfying some algebraic conditions. A thorough discussion on this topic is provided in \cite{TACBOROSCH}. Noteworthy, the well-defined structure of some physical systems permits finding $\gamma(x)$ without solving PDEs, as is shown in Section \ref{sec:5}.
\section{Particular cases}\label{sec:5}
The controllers developed in Sections \ref{sec:3} and \ref{sec:4} are devised to stabilize a rather general class of nonlinear systems characterized by Assumptions \ref{ass1}--\ref{ass3}. In principle, such assumptions should be checked system by system. However, in some particular cases of interest, these assumptions always hold or can be straightforwardly verified. This section focuses on mechanical systems modeled in the pH framework and electrical circuits represented via the BM equations and how the mentioned assumptions are translated into these systems. We stress that these modeling approaches encompass a broad range of systems. See, \cite{BRAYI,BRAYII,GEObook,jeltsema2004tuning,VANJEL}.

\subsection{Mechanical systems in the pH representation}
Consider a mechanical system represented by
\begin{equation}
\arraycolsep=0.8pt
\def\arraystretch{1.5}
\begin{array}{rcl}
 \begin{bmatrix}
  \dot{q} \\[0.2cm] \dot{p}
 \end{bmatrix}&=&
\begin{bmatrix}
 \mathbf{0} & I_{n} \\[0.2cm] -I_n & -\mathcal{D}(q,p)
\end{bmatrix}
\begin{bmatrix}
 \nabla_{q}H(q,p) \\[0.2cm] \nabla_{p}H(q,p) 
\end{bmatrix}
+\begin{bmatrix}
 \mathbf{0} \\[0.2cm] G 
 \end{bmatrix}
u, \\[0.5cm]
H(q,p)&:=&\dfrac{1}{2}p^{\top}M^{-1}(q)p+V(q),
\end{array} \label{phmec}
\end{equation}
where:\footnote{To simplify the notation, we consider that mechanical systems have dimension $2n$.}
\begin{itemize}
 \item $q,p\in\rea^{n}$ represent the generalized positions and momenta, respectively.
 \item $V:\rea^n\to\rea_{+}$ denotes the potential energy of the system.
 \item The so-called inertia matrix $M:\rea^n\to\rea^{n\times n}$ is positive definite. For further details on the computation and properties of this matrix
we refer the reader to \cite{KELROB,SPONGVID}.
\item The Hamiltonian $H:\rea^n\times\rea^n\to\rea_{+}$ is given by the total energy of the system.
\item The input matrix $G\in \rea^{n\times m}$ is of the form
\begin{equation*}
 G = \begin{bmatrix}
      \mathbf{0} \\[0.1cm] I_{m}
     \end{bmatrix}.
\end{equation*} 
\item $\mathcal{D}:\rea^{n}\times \rea^{n}\to \rea^{n\times n}$ is a diagonal positive semi-definite matrix that represents the dissipation---damping---of the system.
\end{itemize}
For mechanical systems of the form \eqref{phmec}, the set of assignable equilibria is 
\begin{equation}
 \mathcal{E}_{\mathcal{M}}:=\left\lbrace (q,p)\in\rea^{n}\times\rea^{n}| G^{\top}\nabla V(q) = \mathbf{0}, p=\mathbf{0}\right\rbrace. \label{eqmec}
\end{equation} 
Given \eqref{phmec}, we make the following observations:
\begin{itemize}
 \item [\textbf{O1}] The system \eqref{phmec} admits a representation of the form \eqref{afsys}, with
 \begin{equation*}
   f(x)= \begin{bmatrix}
 \mathbf{0} & I_{n} \\[0.2cm] -I_n & -\mathcal{D}(q,p)
\end{bmatrix}
\begin{bmatrix}
 \nabla_{q}H(q,p) \\[0.2cm] \nabla_{p}H(q,p) 
\end{bmatrix}, \; g(x)= \begin{bmatrix}
 \mathbf{0} \\[0.2cm] G 
 \end{bmatrix}.
 \end{equation*} 
 \item[\textbf{O2}] Some simple computations show that 
 \begin{equation*}
   \dot{H}= - \lVert \dot{q} \rVert^{2}_{\mathcal{D}(q,p)} + \dot{q}^{\top}Gu.
 \end{equation*} 
 Hence, Assumption \ref{ass1} holds for $S(x)=H(q,p)$, and $\lVert \ell(x) \rVert^{2} = \lVert \dot{q} \rVert^{2}_{\mathcal{D}(q,p)}.$ Moreover, the passive output is given by $y=G^{\top}\dot{q}$, and a suitable selection of $\gamma(x)$ is
 \begin{equation*}
  \gamma(q) = G^{\top}q.
 \end{equation*} 
  \item[\textbf{O3}] In this case,
 \begin{equation*}
  \kappa = -G^{\top}\left( \nabla V \right)_{*}.
 \end{equation*} 
 Then, from \eqref{eqmec}, we have
 \begin{equation*}
  \left( \nabla H \right)_{*} + (\nabla \gamma)_{*}\kappa = \mathbf{0}.
 \end{equation*} 
 \item[\textbf{O4}] Since $\ell_{*}=\mathbf{0}$, there is no dissipation obstacle.
 \item[\textbf{O5}] Since $\mathcal{D}(q,p)$ is diagonal, we can rewrite it as follows
 \begin{equation*}
  \mathcal{D}(q,p) = \bkd{\mathcal{D}_{u}(q,p), \mathcal{D}_{a}(q,p)},
 \end{equation*} 
 where $\mathcal{D}_{u}:\rea^{n}\times\rea^{n}\to\rea^{(n-m)\times (n-m)}$ and $\mathcal{D}_{a}:\rea^{n}\times\rea^{n}\to\rea^{m\times m}$ are diagonal matrices. Accordingly, if $\mathcal{D}_{a}(q,p)$ has full rank and
 $\mathcal{D}_{u}(q,p)$ has at least one nonzero entry, Assumption \ref{ass3} holds with $\Lambda_{c}(x)= \mathcal{D}_{a}(q,p)$, the diagonal matrix $\Lambda_{\ell}(x)$ consists of all the nonzero entries of $\mathcal{D}_{u}(q,p)$,
 and $\eta(x)$ is given by the positions that satisfy $q_{j}\mathcal{D}_{u_{j}}(q,p)\neq \mathbf{0}$, where
 \begin{equation*}
 \begin{array}{rcl}
  q_{j}:=e_{j}^{\top}q, & & \mathcal{D}_{u_{j}}(q,p):=e_{j}^{\top}\mathcal{D}_{u}(q,p)e_{j},
 \end{array}
 \end{equation*}
 for $j=1,\cdots, n-m$.
\end{itemize}
From the observations listed above, we conclude that for mechanical systems that can be expressed as in \eqref{phmec}, the controllers developed in Sections \ref{sec:3} and \ref{sec:4} stabilize
the system at the desired equilibrium if
\begin{equation}
 \begin{array}{l}
  \left(\nabla^2 V \right)_{*}+G\left(\diag\{\alpha_{1}\beta_{1}, \dots, \alpha_{m}\beta_{m} \}\right)G^{\top}>0,\\
   \left(\mathcal{D}+GG^{\top}\right)\dot{q} = \mathbf{0} \implies \left\lbrace \begin{array}{rcl}
                                                                        q&=& q_{*} \\ p&=& \mathbf{0}.
                                                                       \end{array}\right.
 \end{array}\label{meccond}
\end{equation} 
\subsubsection{Fully actuated mechanical systems}
A mechanical system such that $n=m$ is said to be \textit{fully actuated}. This subclass of mechanical systems is of great interest in robotics as a broad range of robotic arms satisfies the aforementioned conditions.\\[0.2cm]
When dealing with fully actuated mechanical systems, it is possible to modify Proposition \ref{pro2} to provide a stronger result, i.e., a saturated controller that guarantees the global asymptotic stability of the desired equilibrium while avoiding velocity measurements. This controller is introduced in the following proposition.
\begin{proposition}\label{pro:fully}\em
 Consider the system \eqref{phmec}, with $G=I_{n}$, the function $\Phi_{c}(\gamma(x),x_{c})$ given in \eqref{Phic}, with $\gamma(x)=q$, and the dynamics of $x_{c}$ provided in \eqref{dxc}. Assume that $\nabla V(q)$ is bounded. Then, the saturated control law
 \begin{equation}
  u=\nabla V(q)- \nabla_{q}\Phi_{c}(q,x_{c}) ,
  \label{umec}
 \end{equation} 
 ensures that $(q_{*},\mathbf{0},\mathbf{0})$ is a \textit{globally asymptotically stable} 
 equilibrium point of the closed-loop system with Lyapunov function
 \begin{equation}
  H_{d}(q,p,x_{c})=\Phi_{c}(q,x_{c})+\frac{1}{2}p^{\top}M^{-1}(q)p+\frac{1}{2}x_{c}^{\top}K_{c}x_{c}. \label{Hdmec}
 \end{equation} 
\end{proposition}
\begin{proof}
 Note that the closed-loop system takes the form
 \begin{equation*}
  \begin{bmatrix}
   \dot{q} \\[0.2cm] \dot{p} \\[0.2cm] \dot{x}_{c}
  \end{bmatrix} = \begin{bmatrix}
                \mathbf{0} && I_{n} && \mathbf{0}\\[0.2cm] -I_n && -\mathcal{D} && \mathbf{0}\\[0.2cm] \mathbf{0}
                && \mathbf{0} && -R_{c}
                  \end{bmatrix}
\begin{bmatrix}
 \nabla_{q} H_{d} \\[0.2cm] \nabla_{p} H_{d} \\[0.2cm] \nabla_{x_{c}} H_{d}
\end{bmatrix}
 \end{equation*}
 with $H_{d}$ defined in \eqref{Hdmec}. Therefore
 \begin{equation*}
  \dot{H}_{d} = -\lVert \dot{q} \rVert^{2}_{\mathcal{D}} - \lVert \nabla_{x_{c}} H_{d} \rVert^{2}_{R_{c}}.
 \end{equation*} 
 Furthermore, some simple computations show that $\left( \nabla H_{d} \right)_{*}=\mathbf{0}$ and $\left( \nabla^{2} H_{d} \right)_{*}>0$ for any $K_{c}>0$. 
Accordingly, $(q_{*},\mathbf{0},\mathbf{0})$ is a stable equilibrium point for the closed-loop system.\\
To prove asymptotic stability, note that, following the arguments given in the proof of Proposition \ref{pro2}, we have that $\dot{H}_{d}=0$ implies $\dot{x}_{c}=\mathbf{0}$ and $y=\dot{q}=\mathbf{0}$. In particular, the latter leads to
\begin{equation*}
 p=\dot{p}=\nabla_{q}\Phi_{c}=\mathbf{0}.
\end{equation*} 
Moreover, since $\nabla_{x_{c}}\Phi_{c}=\nabla_{q}\Phi_{c}$ and $K_{c}$ is full rank, we get that $\dot{x}_{c}=\mathbf{0}$ implies $x_{c}=\mathbf{0}$. Hence,
\begin{equation*}
 \nabla_{q}\Phi_{c}=\displaystyle\sum_{i=1}^{n}e_{i}\alpha_{c_{i}}\tanh(\beta_{c_{i}}(q-q_{*} + x_{c}))=\mathbf{0}
\end{equation*} 
implies $q=q_{*}$.
%
%
The proof is completed noting that $H_{d}(q,p,x_{c})$ is radially unbounded.
\end{proof}
\begin{remark}
 As a result of the comparison between the controllers \eqref{uc} and \eqref{umec}, we note that in the latter, the term $-\kappa$ is replaced with $\nabla V(q)$. 
 The physical interpretation of this is that the controller is canceling the effect of the open-loop potential energy while assigning a new potential energy function 
 with a minimum at the desired position. An example of this is the gravity compensation in robotic arms.  
\end{remark}
\subsubsection{Removing the steady-state error}\label{sec:rse}

Due to the complexity of their characterization, some nonlinear phenomena, e.g., static friction and asymmetry in the motors, are often neglected in the mathematical model of a mechanical system. This may affect the behavior of the closed-loop system. In particular, steady-state errors may arise. A common practice to deal with this problem is adding an integrator of the position error or a filter. However, it is necessary to ensure that the integrator--or filter--does not jeopardize the stability of the closed-loop system. Some solutions to this problem involve a change of coordinates. See, for instance, \cite{dirksz2012power,DONJUN,FERGUSONlhmnc}. However, this may lead to controllers that depend implicitly on the velocities.
%
%
%
Here, we provide a condition that is sufficient to ensure that the addition of the filter that deals with the steady-state error does not affect the stability properties of the closed-loop system. The stability analysis presented below is a direct application of the so-called Lyapunov's indirect method, see \cite{KHA}. While this result is only local, it provides a simple way to ensure the stability of the closed-loop system after the addition of the integrator or filter.\\
Let $\psi\in\rea^{m}$ be the state of the filter, and $\MAP{f_{\psi}}{m}{m}$, $\MAP{\Psi}{m}{m\times m}$ be differentiable functions. Consider the error $\tilde q:=q-q_{*}$, a filter with dyamics
\begin{equation}
\dot{\psi} = f_{\psi}(\psi)  + \Psi(\psi)G^{\top}\tilde{q}, \label{dpsi}
\end{equation} 
and the augmented state vector $\zeta:=(q,p,x_{c},x_{\ell},\psi)$. Hence,
\begin{equation}
 \dot{\zeta} = f_{\zeta}(\zeta,u):=\begin{bmatrix}
M^{-1}(q)p \\[0.2cm] -\nabla_{q}H(q,p)-D(q,p)M^{-1}(q)p + Gu \\[0.2cm] -R_{c}\left( K_{c}x_{c} + \nabla_{x_{c}}\Phi_{c}(z_{c})\right) \\[0.2cm] -R_{\ell}\nabla_{x_{\ell}}\Phi_{\ell}(z_{\ell})  \\[0.2cm]  f_{\psi}(\psi)  + \Psi(\psi)G^{\top}\tilde{q}                                                                                                               
\end{bmatrix} \label{augsys}
\end{equation} 
where $H(q,p)$, $\Phi_{c}(z_{c})$, and $\Phi_{\ell}(z_{\ell})$ are defined in \eqref{phmec}, \eqref{Phic}, and \eqref{Phil}, respectively. Proposition \ref{pro:int} establishes a condition under which we can ensure that the addition of the filter \eqref{dpsi} does not affect the stability of the closed-loop system.
\begin{proposition}\label{pro:int}
 Let $\zeta_{*}=(q_{*},\mathbf{0})$ be the desired equilibrium point for \eqref{augsys}, and consider the augmented system in closed-loop with
\begin{equation}
 \begin{array}{rcl}
u&=&-\displaystyle\sum_{i=1}^{n}e_{i}\left[\alpha_{c_{i}}\tanh\left( \beta_{c_{i}}\left( \tilde{q}+x_{c} \right) \right)+\alpha_{\ell_{i}}\tanh(\beta_{\ell_{i}}z_{\ell_{i}})\right]\\[0.2cm]&& +u_{\psi}(\psi)+G^{\top}\left( \nabla_{q}V \right)_{*}, \end{array} \label{conint}
\end{equation}
yielding the closed-loop dynamics $\dot{\zeta}=f_{\zeta_{cl}}(\zeta)$. Then, $\zeta_{*}$ is a locally asymptotically stable equilibrium for the closed-loop system if the linear system
\begin{equation*}
 \dot{\zeta} = \left(\nabla f_{\zeta_{cl}}\right)_{*}\tilde\zeta
\end{equation*} 
is stable, where $\tilde\zeta:=\zeta-\zeta_{*}$.
\end{proposition}
\begin{proof}
 The proof follows from Lyapunov's indirect (linearization) method. For further details see Theorem 4.7 in \cite{KHA}.
\end{proof}
At this point we make the following observations:
\begin{itemize}
 \item [\textbf{O6}] If $f_{\psi}=\mathbf{0}$ and $\Psi(\psi) = I_{m}$, then \eqref{dpsi} is an integrator. Furthermore, by fixing $u_{\psi}(\psi) = -\psi$, we obtain a classical integrator of the position error. 
 \item [\textbf{O7}] If $u_{\psi}(\psi)$ is saturated and $\left( \nabla_{q}V \right)_{*}$ is bounded, then the controller \eqref{conint} is saturated. Thus, a natural choice for $u_{\psi}$ is
 \begin{equation*}
  u_{\psi}(\psi)= -\displaystyle\sum_{i=1}^{m}e_{i}\alpha_{\psi_{i}}\tanh(\beta_{\psi_{i}}\psi_{i}),
 \end{equation*}
 where $\alpha_{\psi_{i}}$ and $\beta_{\psi_{i}}$ are positive constants.
\end{itemize}
Note that Proposition \ref{pro:int} can be straightforwardly adapted to the case when $x_{\ell}$ is not necessary, for instance, for fully actuated mechanical systems.  

\subsection{Electrical circuits}
In general, the pH approach is suitable to model 
electrical circuits composed of passive components. Nevertheless, the state variables of such models are fluxes and charges, which most of the time are not measurable signals. 
A solution to this problem is to represent the behavior of the electrical networks via the BM equations, where the state variables are voltages and currents. In this section,
we study the BM equations that represent a broad class of electrical networks. Then, we provide sufficient conditions to ensure that the controllers developed in Sections \ref{sec:3} and \ref{sec:4} are suitable for stabilizing these systems.\\[0.2cm]
We restrict our attention to \textit{topologically complete} networks. Accordingly, below we introduce this definition, which is taken from \cite{WEISS}. Then, we refer the reader to the mentioned reference for further discussion on topologically complete networks and related literature.

\begin{definition}
A topologically complete network of two terminal voltage-controlled and current-controlled elements has a graph that possesses a tree containing all of the capacitive branches and none of the inductive branches, each resistive tree branch corresponds to a current-controlled resistor, each resistive link corresponds to a voltage-controlled resistor, and finally, for which the location of the resistive branches are such that there exist no fundamental loop in which resistive branches appear both as tree branches and as links.
\end{definition}

Consider an electrical network consisting of $\varsigma$ linear inductors, $\varpi$ linear capacitors, and none controlled nor constant source. Then, $i_{L}: =\left[ i_{L_{1}}, \cdots, i_{L_{\varsigma}} \right]^{\top}\in\rea^{\varsigma}$
represents the currents through the inductors, and $v_{C}:= \left[ v_{C_{1}},\cdots, v_{C_{\varpi}}\right]^{\top}\in\rea^{\varpi}$ denotes the voltages across the capacitors, where $\varsigma + \varpi = n$. Hence, the electrical network can be represented by BM equations, see \cite{BRAYI,BRAYII}, as
 \begin{equation}
  \begin{array}{rcl}
   -L\dfrac{di_{L}}{dt} &=& \nabla_{i_{L}}P(i_{L},v_{C})+\tilde{g}_{L}u_{L} \\[0.3cm] 
   C\dfrac{dv_{C}}{dt} &=& \nabla_{v_{C}}P(i_{L},v_{C})+\tilde{g}_{C}u_{C}
  \end{array}\label{BMe}
 \end{equation} 
 where the positive definite matrices $L\in \rea^{\varsigma\times \varsigma}$ and $C\in\rea^{\varpi\times\varpi}$ denote the inductance and capacitance matrices, respectively, and the so-called \textit{mixed-potential} function is given by
 \begin{equation}
  P(i_{L},v_{C}):=i_{L}^{\top}\Gamma v_{C} + P_{R}(i_{L}) - P_{G}(v_{C}),
  \label{POT}
 \end{equation}
 where:
 \begin{itemize}
  \item The matrix $\Gamma\in\rea^{\varsigma\times\varpi}$ determines the interconnection between the inductors and capacitors of the system, and all its entries are either $1$, $-1$, or zero.
  \item The mappings $\MAP{P_{R}}{\varsigma}{}$ and $\MAP{P_{G}}{\varpi}{}$ are the \textit{dissipative current-potential} and the \textit{dissipative voltage-potential}, respectively, with\footnote{For a more detailed explanation about the mixed-potential function \eqref{POT}-\eqref{DISPOT}, we refer the reader to \cite{jeltsema2003passivity,jeltsema2004tuning}.}
 \begin{equation}
  \begin{array}{rcl}
   P_{R}(i_{L})&:=&\displaystyle\int_{0}^{i_{L}}v_{R}(i_{L}')di_{L}'\\[0.4cm]
   P_{G}(v_{C})&:=&\displaystyle\int_{0}^{v_{C}}i_{G}(v_{C}')dv_{C}',
  \end{array}\label{DISPOT}
 \end{equation} 
 where $v_{R}(0)=\mathbf{0}$, $i_{G}(0)=\mathbf{0}$, and
 \begin{equation}
  \begin{array}{rclrcll}
   i_{L}^{\top}v_{R}(i_{L})&\geq& 0, & \nabla v_{R}(i_{L})&\geq& 0 & \forall \ i_{L}\in\rea^{\varsigma} \\[0.2cm]
   v_{C}^{\top}i_{G}(v_{C})&\geq& 0, & \nabla i_{G}(v_{C})&\geq& 0 & \forall \ v_{C}\in\rea^{\varpi}.
  \end{array}\label{disRLC}
 \end{equation} 
 \item The inputs $u_{L}\in\rea^{m_{\varsigma}}, \ u_{C}\in\rea^{m_{\varpi}}$, with $m_{\varsigma}+m_{\varpi}=m$, denote the external voltage sources in series with the inductors and the external current sources in parallel with the capacitors, respectively.
\item The constant matrices $\tilde{g}_{L}\in\rea^{\varsigma\times m_{\varsigma}}, \ \tilde{g}_{C}\in\rea^{\varpi\times m_{\varpi}}$ represent the \textit{voltage-related} input matrix and the \textit{current-related} input matrix, respectively.
 \end{itemize} 
 The system \eqref{BMe}-\eqref{POT} admits a more compact representation of the form
 \begin{equation}
  Q\dot{x} = \nabla P(x) + \tilde{g}u \label{BMrlc}
 \end{equation} 
 with
 \begin{equation}
  \begin{array}{rcl}
   x=\begin{bmatrix}
      i_{L} \\[0.1cm] v_{C}
     \end{bmatrix},&
Q:= \bkd{-L,C},&\tilde{g}:=\begin{bmatrix}
\tilde{g}_{L} \\[0.1cm] \tilde{g}_{C}                        
                                   \end{bmatrix}
.
  \end{array}\label{xQg}
 \end{equation}
 Note that the system \eqref{BMrlc} can be expressed as in \eqref{afsys} with
 \begin{equation}
  \begin{array}{rcl}
   f(x) = Q^{-1}\nabla P(x), && g = Q^{-1}\tilde g.  \label{fgRLC}
  \end{array}
 \end{equation} 
 Customarily, to ensure that Assumption \ref{ass1} holds, we look for an alternative pair $(\tilde{Q}(x),\tilde{P}(x))$, for a detailed discussion on this topic we refer the reader to \cite{jeltsema2005brayton} and \cite{BORSCHERCDC18}.
 In particular, for the systems under study in this section we have the following result.\\[0.1cm]
 \begin{proposition}\label{pro:rcl}
  Suppose that the Hessian of $P(x)$ has full rank. Define
  \begin{equation}
   \begin{array}{rcl}
    \tilde{P}(x) &:=& \displaystyle\frac{1}{2}\left[ \nabla P(x) \right]^{\top}\Xi\nabla P(x) \\[0.2cm]
    \tilde{Q}(x) &:=& \nabla^{2}P(x)\Xi Q
   \end{array} \label{PQtilde}
  \end{equation} 
 \end{proposition}
 where $P(x)$ and $Q$ are defined in \eqref{POT}--\eqref{DISPOT} and \eqref{xQg}, respectively, and
  \begin{equation*}
   \Xi:=\bkd{L^{-1},C^{-1}}.
  \end{equation*}
  Then, the system \eqref{BMrlc} can be rewritten as
  \begin{equation}
   \dot{x} = \tilde{Q}^{-1}(x)\nabla \tilde{P}(x) + gu. \label{RLCaff}
  \end{equation} 
  Furthermore, the map $u\mapsto -g^{\top}\tilde{Q}(x)^{-\top}\dot{x}$ is passive with storage function $\tilde{P}(x)$.
 \begin{proof}
Note that 
\begin{equation*}
 \nabla \tilde{P} = \nabla^{2}P(x)\Xi\nabla P. 
\end{equation*} 
 Hence, 
 \begin{equation*}
  Q^{-1}\nabla P(x) = \tilde Q^{-1}(x)\nabla \tilde P(x),
 \end{equation*} 
 and the expression \eqref{RLCaff} is obtained from \eqref{fgRLC}. To prove passivity note that
 \begin{equation*}
   \tilde{Q}(x) 
   =\begin{bmatrix}
    -\nabla v_{R}(x) & \Gamma \\ -\Gamma^{\top} & -\nabla i_{G}(x) 
   \end{bmatrix}.
 \end{equation*} 
 Thus, from \eqref{disRLC}, the symmetric part of $\tilde{Q}(x)$ is negative semi-definite. Hence, by premultiplying both sides of \eqref{RLCaff} by $\dot{x}^{\top}\tilde{Q}(x)$, we obtain  
  \begin{equation*}
    \dot{x}^{\top}\tilde Q(x)\dot{x} = \dot{\tilde P} + \dot{x}^{\top}\tilde Q(x) gu  \implies \dot{\tilde P} \leq  -\dot{x}^{\top}\tilde Q(x) gu.
  \end{equation*} 
 \end{proof}
 In light of Proposition \ref{pro:rcl}, we make the following observations:
 \begin{itemize}
  \item [\textbf{O8}] Assumption \ref{ass1} holds for $S(x) = \tilde P(x)$.
  \item [\textbf{O9}] The components of $y$ are given in terms of the dynamics of the system, which might be non-measurable signals. On the other hand, the elements of $\gamma(x)$ can be expressed in terms of voltages and currents, which are, in general, the available measurements in an RLC network.
  \item [\textbf{O10}] The asymptotic stability of the equilibrium in Propositions \ref{pro1}--\ref{pro4} is ensured if
  \begin{equation*}
   \diag\left\lbrace \nabla v_{R}(i_{L}), \nabla i_{G}(v_{C}) \right\rbrace\dot{x} = \mathbf{0} \implies x=x_{*}. 
  \end{equation*} 
 \end{itemize}
 
We conclude this section with the following remark concerning the integrability of the passive output provided in Proposition \ref{pro:rcl}.

 \begin{remark}
 As a result of the Poincar\'{e} lemma, there exists $\gamma(x)$ such that $\dot{\gamma} = y$ if 
  \begin{equation*}
   \nabla(\tilde Q(x)g(x)) = \left[ \nabla(\tilde Q(x)g(x)) \right]^{\top}. 
  \end{equation*}
 \end{remark}

\section{Examples}\label{sec:6}
In this section, we illustrate the applicability of the controllers presented in Sections \ref{sec:3}, \ref{sec:4} and \ref{sec:5} through the stabilization of three systems in different physical domains. To this end, we present simulations and experimental results derived from the implementation of the mentioned controllers.
\subsection{Electromechanical (translational) coupling device}
\begin{figure}[h]
 \centering
 \includegraphics[width=0.8\textwidth]{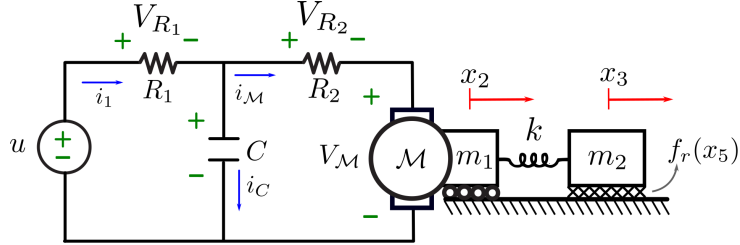}
 \caption{Electromechanical coupling device.}
 \label{fig:coupling}
\end{figure}

Consider the coupling device depicted in Fig. \ref{fig:coupling}, where $u$ is the voltage provided by the source, $R_{1}$ and $R_{2}$ denote linear resistors, $C$ represents a linear capacitor, the electrical part of the system is \textit{coupled} with the mechanical one via the motor $\mathcal{M}$, the symbol $k$ represents a linear spring, $x_{1}$ is the charge across the capacitor, $x_{2}$ and $x_{3}$ are the positions of the masses, $x_{4}$ and $x_{5}$ are the momenta. The term $f_{r}(x_{5})$ is an approximation of the friction force present in the second mass, which is given by
\begin{equation*}
 f_{r}(x_{5}) = \frac{a_{1}}{m_{2}}x_{5} + a_{2}\tanh(a_{3}x_{5}),
\end{equation*} 
where $a_{1}$, $a_{2}$, and $a_{3}$ are positive constant parameters. Hence, the dynamics of this system can be represented as in \eqref{afsys}, with $g=e_{1}\frac{1}{R_{1}}$ and
\begin{equation*}
  f(x) = \begin{bmatrix}
          -\displaystyle\left( \frac{1}{R_{1}}+\frac{1}{R_{2}} \right)\frac{1}{C}x_{1} + \frac{1}{a_{0}m_{1}R_{2}}x_{4} \\[0.4cm]
          \displaystyle\frac{1}{m_{1}}x_{4} \\[0.4cm] \displaystyle\frac{1}{m_{2}}x_{5} \\[0.4cm] -k\left( x_{2} - x_{3} \right) + \displaystyle\frac{1}{a_{0}R_{2}}\left( \frac{1}{C}x_{1}- \frac{1}{a_{0}m_{1}}x_{4}\right) \\[0.4cm]  k\left( x_{2} - x_{3} \right) - \displaystyle\frac{a_{1}}{m_{2}}x_{5}-a_{2}\tanh(a_{3}x_{5})
         \end{bmatrix}, 
\end{equation*} 
where $a_{0}$ is a positive constant parameter that characterizes the relation between the electrical and mechanical variables of the motor.\\
The set of assignable equilibria for this system is given by
\begin{equation*}
 \mathcal{E} = \left\lbrace x\in\rea^{5}\mid x_{1}=x_{4}=x_{5}=0, x_{2}=x_{3} \right\rbrace,
\end{equation*} 
and the control objective is to stabilize the mass $m_{2}$ at the desired point $x_{3_{*}}$ while considering that the voltage source has a limited operation range, and there are no sensors to measure velocities. To this end, consider the total energy of the system, given by,
\begin{equation*}
 S(x) = \displaystyle\frac{1}{2C}x_{1}^{2} + \displaystyle\frac{1}{2}k\left( x_{2}-x_{3} \right)^{2}  + \displaystyle\frac{1}{2m_{1}}x_{4}^{2} + \displaystyle\frac{1}{2m_{2}}x_{5}^{2}. 
\end{equation*} 
Then, some simple computations show that 
\begin{equation*}
\begin{array}{rcl}
 \left[ \nabla S(x) \right]^{\top}f(x) &=& -\displaystyle\frac{1}{C^{2}R_{1}}x_{1}^{2}-\displaystyle\frac{1}{R_{2}}\left( \frac{1}{C}x_{1} - \frac{1}{a_{0}m_{1}}x_{4} \right)^{2} - \dfrac{a_{1}}{m_{2}^{2}}x_{5}^{2} - \dfrac{a_{2}}{m_{2}}x_{5}\tanh(a_{3}x_{5}).
 \end{array}
\end{equation*} 
Accordingly, Assumption \ref{ass1} is satisfied. Moreover, 
\begin{equation}
 \begin{array}{rcl}
\dot{S} &=& -\displaystyle\frac{R_{1}R_{2}}{R_{1}+R_{2}}\dot{x}_{1}^{2}-\displaystyle\frac{1}{a_{0}^{2}\left( R_{1}+R_{2} \right)}\dot{x}_{2}^{2} - \left[a_{1}\dot{x}_{3}+a_{2}\tanh(m_{2}a_{3}\dot{x}_{3})\right]\dot{x}_{3}+y^{\top}u \leq y^{\top}u \end{array} \label{dSe1}
\end{equation} 
with
\begin{equation*}
 y = \displaystyle\frac{R_{2}}{R_{1}+R_{2}}\dot{x}_{1} + \displaystyle\frac{1}{a_{0}\left(R_{1} + R_{2}\right)}\dot{x}_{2}.
\end{equation*} 
Therefore, 
\begin{equation}
 \gamma(x) = \displaystyle\frac{R_{2}}{R_{1}+R_{2}}x_{1} + \displaystyle\frac{1}{a_{0}\left(R_{1} + R_{2}\right)}x_{2} \label{gamex1}
\end{equation} 
satisfies $\dot{\gamma} = y$. Furthermore, Assumption \ref{ass2} holds for $\gamma(x)$ given in \eqref{gamex1}, $\kappa = 0$, and any $\alpha,\beta>0$. Note that, given \eqref{dSe1}, Assumption \ref{ass3} is satisfied for $\eta=x_{3}$, where $\Lambda_{c}$ depends on the value of the resistors. For this example, we have $\Upsilon=1$. Hence, from Proposition \ref{pro4}, it follows that the controller \eqref{ulc} ensures that the closed-loop system has a stable equilibrium at $(0,x_{2_{*}},x_{3_{*}}, 0, 0)$, with $x_{2_{*}}=x_{3_{*}}$. We remark that, since $\gamma(x_{1},x_{2})$ and $\eta=x_{3}$, the control law does not depend on $x_{4}$ and $x_{5}$, which are the states related to the velocities of the masses.

To prove asymptotic stability of the equilibrium, we check if \eqref{detectlc} holds. For this example, we have the following chain of implications
\begin{equation}
 \begin{array}{rcccccl}
  &\dot{\eta} = 0 &\iff& \dot{x}_{3} = 0 &\iff& x_{5} = 0 \\
  \implies & \dot{x}_{5}=0 &\iff & x_{2} = x_{3} &\implies& \dot{x}_{2} = \dot{x}_{3} \\
  \implies& \dot{x}_{2} = 0 &\iff& x_{4} = 0 &\implies& \dot{x}_{4}=0 \\\iff& x_{1} = 0 &\implies& \dot{x}_{1} = 0 &\implies& u = 0.
 \end{array} \label{implex1}
\end{equation} 
On the other hand, $\dot{x}_{\ell} = 0$ implies that $z_{\ell} = 0$, which combined with \eqref{implex1} leads to the conclusion $z_{c} = 0$. Therefore,
\begin{equation*}
 \begin{array}{cl}
 & \dot{x}_{c}=0 \iff x_{c} = 0 \implies {\gamma}(x) = \gamma_{*} \iff x_{2} = x_{2_{*}} \\ \implies& x_{3} = x_{3_{*}}.
 \end{array}
\end{equation*} 
Consequently, \eqref{detectlc} holds, and the equilibrium point is asymptotically stable. Furthermore, in this case, $S_{d_{\ell}}(x)$ is \textit{radially unbounded}. Thus, the equilibrium point is globally asymptotically stable.
\subsubsection*{Simulations}
\begin{table}[h]
 \caption{Parameters of the electromechanical coupling device}
 \label{tab:parex1}
 \begin{center}
 \renewcommand{\arraystretch}{1.2}
\begin{tabular}{cc|cc}\hline
Parameter & Value & Parameter & Value \\\hline
$R_{1}$ & $100$ & $R_{2}$ & $100$ \\  $C$ & $2.2\times 10^{-4}$ & $m_{1}$ & $0.01$ \\ $m_{2}$ & $0.015$ &
$a_{0}$ & $0.005$ \\ $a_{1}$ & $6\times 10^{-4}$ & $a_{2}$ & $8\times 10^{-5}$ \\ $a_{3}$ & $40$ & $k$ & $0.3$\\\hline
\end{tabular}
\end{center}
\end{table}
To corroborate the effectiveness of the saturated controller, we perform simulations considering the parameters provided in Table \ref{tab:parex1}, where we are particularly interested in showing that the control signal is saturated and the influence of the term $\nabla_{z_{\ell}}\Phi_{\ell}(z_{\ell})$ on the performance of the closed-loop system. To this end, we consider  $0.0 25[m]$ as the desired displacement for the masses, i.e., $x_{2_{*}}=x_{3_{*}}=0.025$, and the control parameters
\begin{equation}
 \begin{array}{lll}
K_{c} = 10^{6}, & R_{c} =0.3 , & \beta_{c} =450, \\
K_{\ell} = 5.5\times 10^{-4}, & R_{\ell}=33 , & \beta_{\ell}=2\times 10^{6}.
\end{array}\label{cpex1}
\end{equation}

To illustrate how the term $\nabla_{z_{\ell}}\Phi_{\ell}(z_{\ell})$ affects the closed-loop behavior, we consider that the voltage source operates in the range of $\pm 5[V]$.
Fig. \ref{fig:alzero} shows the results of simulating two different scenarios: (i) $\alpha_{c} = 5, \ \alpha_{\ell}= 0$, which is plotted in blue, and (ii) $\alpha_{c} = 2.5, \ \alpha_{\ell}= 2.5$ plotted in orange. In both cases the initial conditions are $\mathbf{0}$. From the plots, we observe that in the scenario (i), the first mass converges towards the desired position without oscillations. In contrast, the second mass exhibits an oscillatory behavior as the natural damping is relatively small. On the other hand, in the second scenario, it is evident that the term $\nabla_{z_{\ell}}\Phi_{\ell}(z_{\ell})$ injects damping to the second mass, reducing notoriously the oscillations in $x_{3}$.

\begin{figure}[h]
 \centering
 \includegraphics[width=\textwidth]{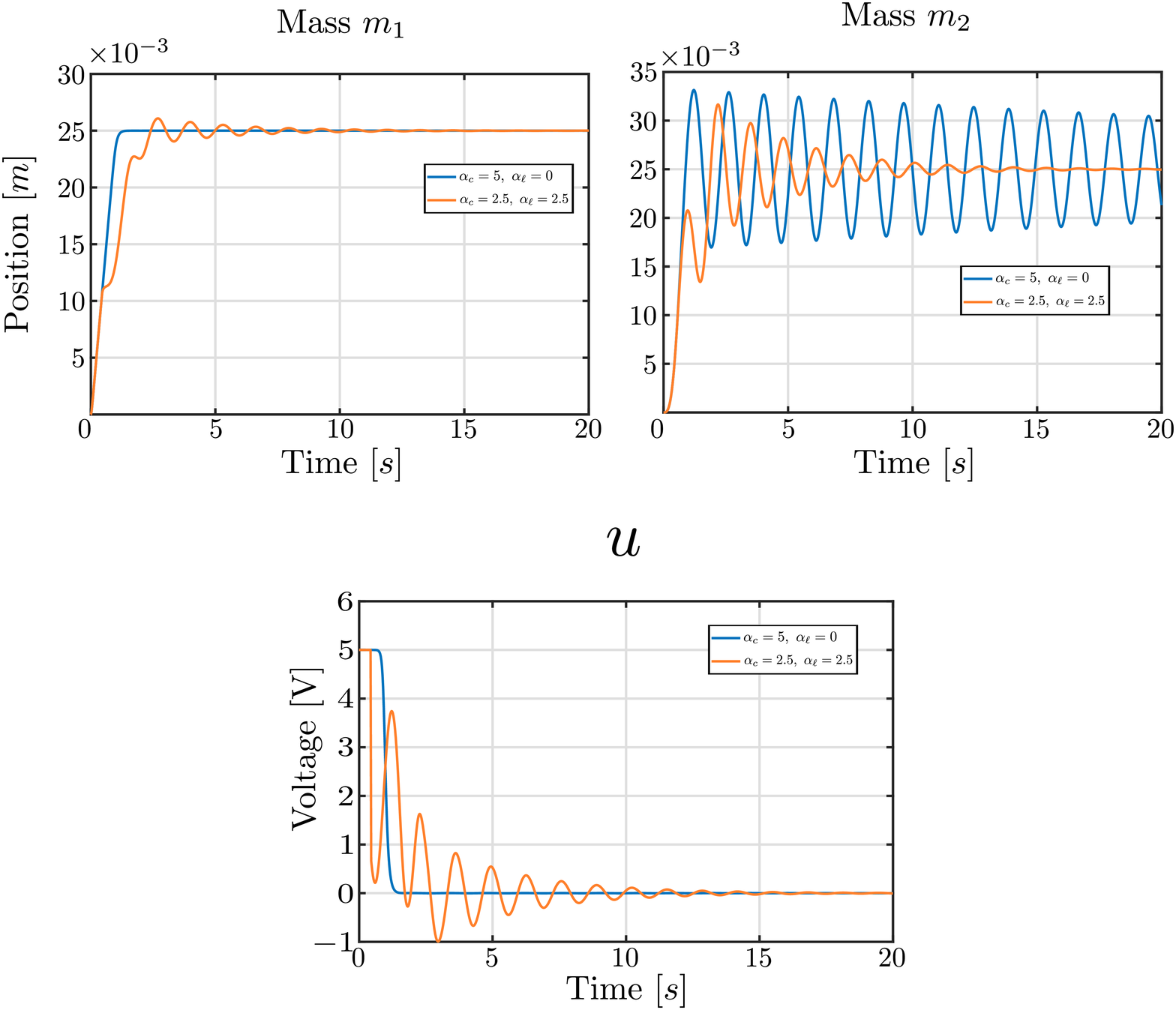}
 \caption{Evolution of the masses and the control law with and without damping injection through the term $\nabla_{z_{\ell}}\Phi_{\ell}(z_{\ell})$.}
 \label{fig:alzero}
\end{figure}

To show the saturation of the controller, we consider the control parameters \eqref{cpex1} and three different set of values for $\alpha_{c}$ and $\alpha_{\ell}$, namely,
\begin{equation*}
 \begin{array}{rclcrcl}
  \alpha_{c}&=&2.5,& \alpha_{\ell}&=&2.5, \\
  \alpha_{c}&=&3.75,& \alpha_{\ell}&=&3.75, \\
  \alpha_{c}&=&5,& \alpha_{\ell}&=&5.
 \end{array}
\end{equation*} 
Accordingly, the corresponding control laws must saturate at $\pm 5, \ \pm 7.5$, and $\pm 10$, respectively. In all cases, we consider initial conditions equal to zero. The simulation results are depicted in Figs. \ref{fig:masses} and \ref{fig:u2sec}. In the former, we note that the behavior of the masses is not drastically affected by the saturation limits. On the other hand, the saturation of the control signals is appreciated in Fig. \ref{fig:u2sec}, where the plot at the right-hand shows the first two seconds of simulation when the saturation takes place. 

\begin{figure}[h]
 \centering
 \includegraphics[width=\textwidth]{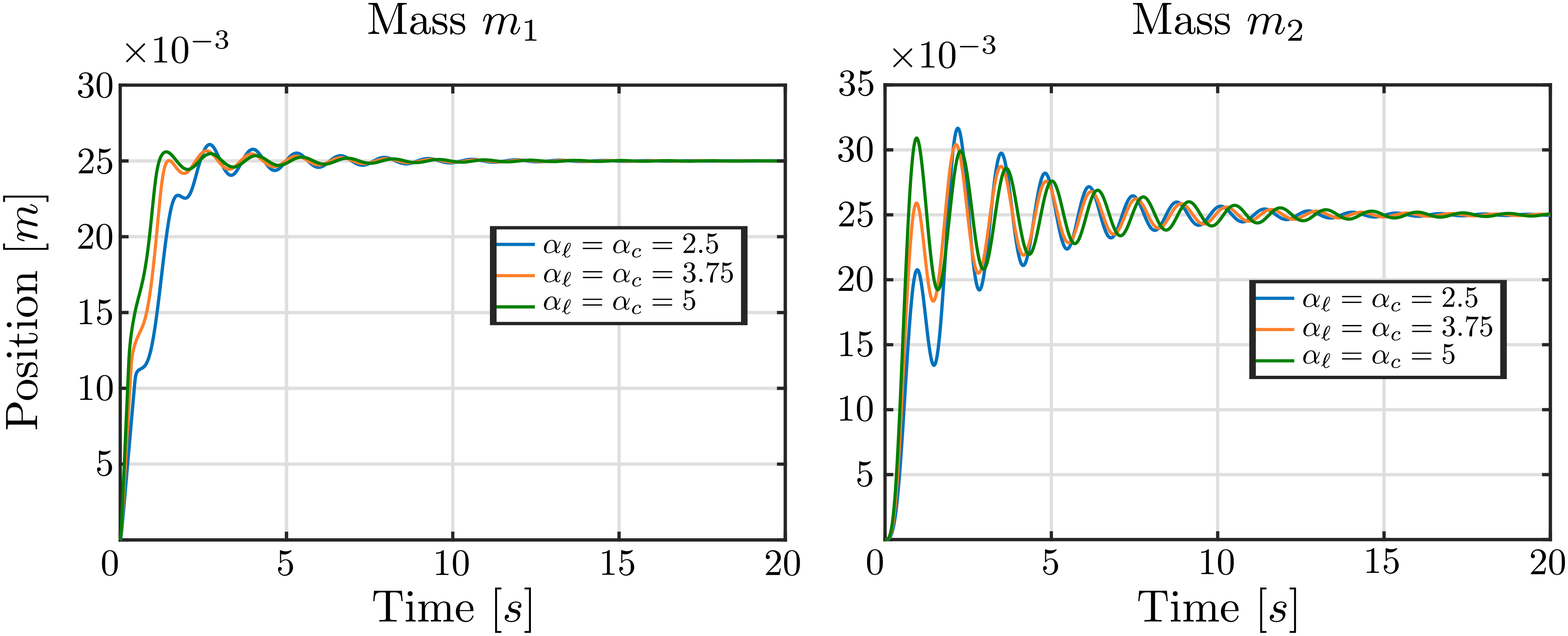}
 \caption{Behavior of the masses for different saturation limits.}
 \label{fig:masses}
\end{figure}
\begin{figure}[h]
 \centering
 \includegraphics[width=\textwidth]{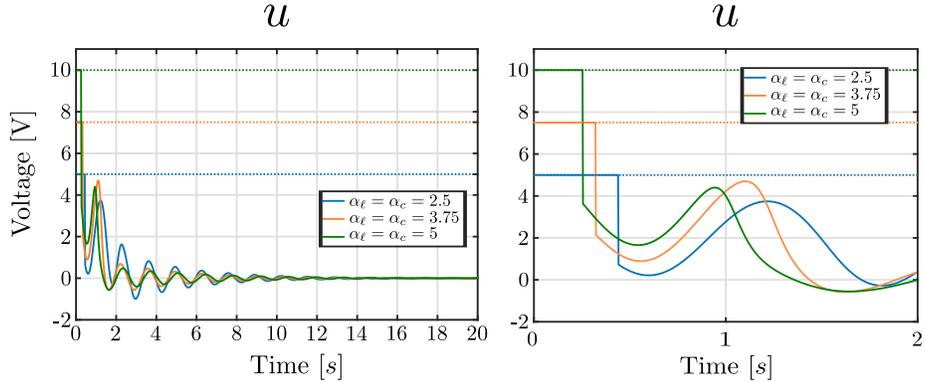}
 \caption{Control signals evolution during the 20 seconds of simulation (left-hand), and only the first 2 seconds (right-hand). The saturation limits for each control signal are plotted with a dotted line.}
 \label{fig:u2sec}
\end{figure}

\subsection{Nonlinear RLC circuit}
\begin{figure}[h]
 \centering
 \includegraphics[width=.8\textwidth]{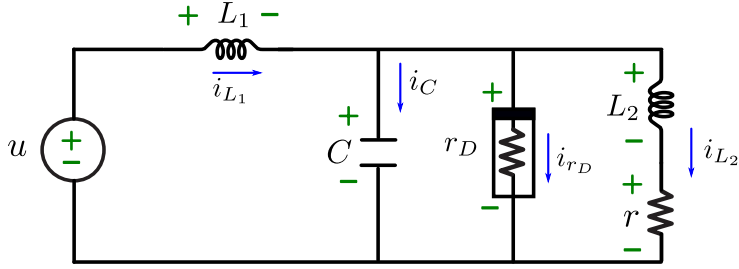}
 \caption{Nonlinear RLC circuit.}
 \label{fig:circuit}
\end{figure}
Consider the circuit depicted in Fig. \ref{fig:circuit}, which admits a representation of the form \eqref{BMrlc} with $\tilde{g}_{C}=0$, $\tilde{g}_{L}=-e_{1}$, and
\begin{equation}
\begin{array}{rclrcl}
 Q&=&\diag\{-L_{1},-L_{2},C \}, & \Gamma &=& \begin{bmatrix}
                                        1 \\ -1
                                       \end{bmatrix}, \\[0.4cm]
i_{r_{D}}(x)&=& -a\left( e^{\frac{x_{3}}{b}} - 1 \right), & v_{R}(x)&=&\begin{bmatrix}
             0 \\ rx_{2}
            \end{bmatrix}, \\[0.4cm] 
\end{array} \label{cir} 
\end{equation}
where $x_{1},x_{2}$ denote the currents through the inductors, $x_{3}$ represents the voltage across the capacitor, the constant parameters $L_{1},L_{2},C$ denote the inductances and the capacitance, respectively, $a$ and $b$ are positive constant parameters, $r_{D}$ is a nonlinear load, and $r$ denotes the resistance of the linear resistor.\\[0.2cm]
The control objective is to regulate the current through $r_{D}$ at the desired value while keeping the supplied voltage bounded. Moreover, we consider that the only measure available is the voltage across the capacitor. To solve this problem, we first define the set of assignable equilibria for this system, which is given by
\begin{equation}
 \mathcal{E}=\left\lbrace x\in\rea^{3}\mid x_{1} = \frac{1}{r}x_{3}+a\left( e^{\frac{x_{3}}{b}} - 1 \right), x_{2} = \frac{1}{r}x_{3}\right\rbrace. \label{eqex2}
\end{equation} 
According to Proposition \ref{pro:rcl}, this system can be represented as in \eqref{RLCaff} with $g=e_{1}\frac{1}{L_{1}}$ and
\begin{equation*}
 \begin{array}{rcl}
  \tilde{Q}(x) &=& \begin{bmatrix}
     0 & 0 & 1 \\ 0 & -r & -1 \\ -1 & 1 & -\frac{a}{b}e^{\frac{x_{3}}{b}}             
                 \end{bmatrix}, \\[0.2cm]
\tilde{P}(x) &=& \frac{1}{2L_{1}}x_{3}^{2}+\frac{1}{2L_{2}}\left( rx_{2}-x_{3} \right)^{2} + \frac{1}{2C}\left[x_{1}-x_{3}+i_{r_{D}}(x) \right]. 
 \end{array}
\end{equation*} 
Moreover, a passive output for this system is given by
\begin{equation*}
 y = \frac{1}{L_{1}}\dot{x}_{3},
\end{equation*} 
with storage function $S(x)=\tilde{P}(x)$. Hence, Assumption \ref{ass1} is satisfied. Furthermore, $\gamma(x)$ can be chosen as
\begin{equation*}
 \gamma(x)=\frac{1}{L_{1}}x_{3},
\end{equation*} 
and some simple computations show that Assumption \ref{ass2} holds for every $\alpha, \ \beta>0$, and $\kappa = -x_{3_{*}}$. Accordingly, from Proposition \ref{pro2} it follows that the control law \eqref{uc} renders stable the equilibrium point $(x_{*},0)$, where $x_{*}\in\mathcal{E}$. Notice that, since $\gamma(x)$ depends exlcusively on $x_{3}$, the controller only requires to measure the voltage across the capacitor.

To prove the asymptotic stability of the equilibrium, note that
\begin{equation*}
 -\lVert \ell(x) + w(x)u \rVert = \dot{x}^{\top}\tilde{Q}(x)\dot{x} = -r\dot{x}_{2}^{2} - \frac{a}{b}e^{\frac{x_{3}}{b}}\dot{x}_{3}^{2}.
\end{equation*} 
Hence, we have the following chain of implications
\begin{equation}
 \begin{array}{rcl}
\dot{S}_{d_{c}}=0&\iff& \left\lbrace\begin{array}{rcl}
                                      \dot{x}_{3}=0\\
                                      \dot{x}_{2}=0
                                     \end{array}\right\rbrace\implies \dot{x}_{1}=0\\ &\implies& x_{3}-x_{3_{*}}=-\alpha_{c}\tanh(\beta_{c}z_{c}).                                                                                                                                              \end{array}
                                     \label{detecex21}
\end{equation} 
On the other hand,
\begin{equation}
\dot{x}_{c} = 0 \iff K_{c}x_{c} = -\alpha_{c}\tanh(\beta_{c}z_{c}). \label{detecex22}  
 \end{equation} 
 Therefore, combining \eqref{detecex21} and \eqref{detecex22}, we conclude that $K_{c}x_{c} = x_{3}-x_{3_{*}}$. Accordingly, we have
 \begin{equation*}
  x_{c} = \displaystyle\frac{L_{1}}{K_{c}+L_{1}}z_{c}.
 \end{equation*} 
 Then, \eqref{detecex22} can be rewritten as
 \begin{equation*}
  \dot{x}_{c} = 0 \iff \displaystyle\frac{K_{c}L_{1}}{K_{c}+L_{1}}z_{c} = -\alpha_{c}\tanh(\beta_{c}z_{c}),
 \end{equation*}
 which holds only if $z_{c}=x_{c}=0$. Moreover, from \eqref{detecex21}, we conclude that
 \begin{equation*}
  x_{3} = x_{3_{*}}\implies\left\lbrace \begin{array}{rcl}
                                         x_{1}&=& x_{1_{*}} \\ x_{2}&=& x_{2_{*}}. 
                                        \end{array} \right.
 \end{equation*} 
\subsubsection*{Simulations}
\begin{table}[h]
 \caption{Parameters of the nonlinear RLC circuit}
 \label{tab:parex2}
 \begin{center}
 \renewcommand{\arraystretch}{1.2}
\begin{tabular}{cc|cc}\hline
Parameter & Value & Parameter & Value\\\hline 
$r$ & $100$ & $L_{1}$ & $0.01$ \\ $L_{2}$ & $0.02$  &
$C$ & $2\times 10^{-4}$ \\  $a$ & $10^{-7}$ & $b$ & $0.25$ \\\hline
\end{tabular}
\end{center}
\end{table}

To corroborate the results exposed above, we consider the parameters given in Table \ref{tab:parex2} and the following scenario: the voltage source must operate in the range of $0$ to $3.1$ volts, and the load demands a current of $20[mA]$. Then, $u_{*}=x_{3_{*}}=3.0515$. Thus, the load drives the voltage source near to its operation limit. Accordingly, we need to ensure that the control signal saturates to protect the source. To this end, we consider the control parameters $K_{c}=10$, $R_{c}=10$, and $\alpha_{c}=0.0485$. Note that the selected value for $\alpha_{c}$ ensures that the control signal saturates at $3.003$ and $3.1$ volts. Fig. \ref{fig:uid} depicts the simulation results considering initial conditions zero and different values for $\beta_{c}$. We observe that larger values for $\beta_{c}$ provoke that the control signal reaches the saturation limits more times since the control law becomes more sensitive to errors between the actual current and the desired one.

\begin{figure}[h]
 \centering
 \includegraphics[width=\textwidth]{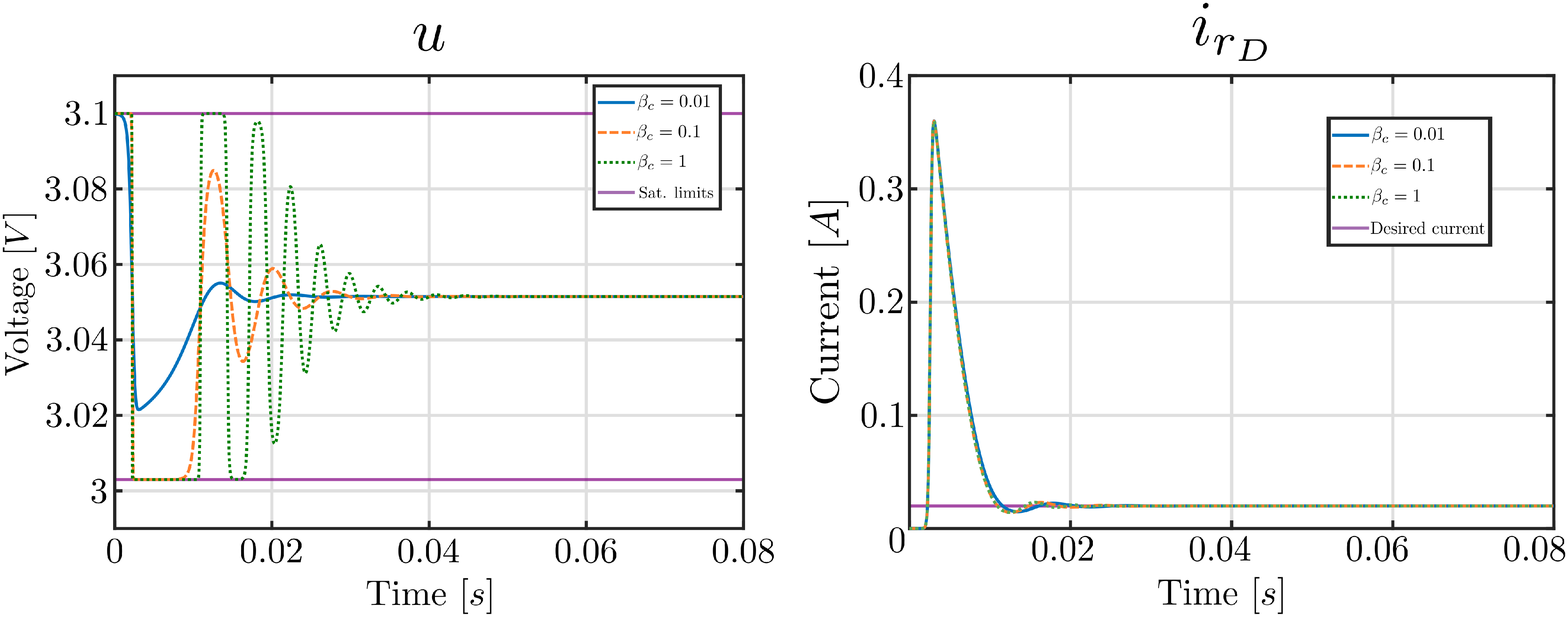}
 \caption{Control signals for different values of $\beta_{c}$ (left-hand), and current through the load (right-hand).}
 \label{fig:uid}
\end{figure}
\subsection{Robotic arm}
\begin{figure}[h]
 \centering
 \includegraphics[width=0.6\textwidth]{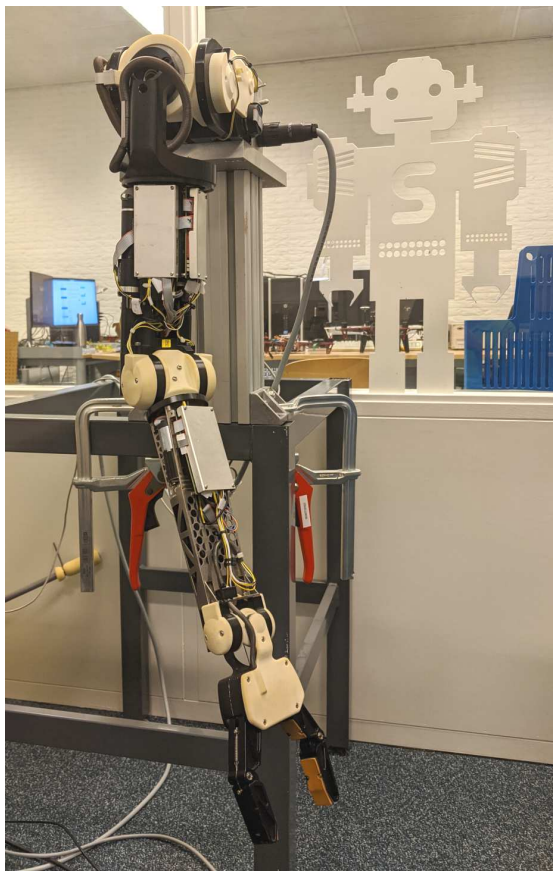}
 \caption{PERA system.}
 \label{fig:PERA}
\end{figure}

Consider the Philips Experimental Robot Arm (PERA) shown in Fig. \ref{fig:PERA}, a robotic arm designed to mimic the human arm motion \cite{rijs2010philips}. To illustrate the applicability of the results reported in Section \ref{sec:3} and \ref{sec:5}, we carry out experiments with the PERA system considering only three degrees-of-freedom, namely, the shoulder roll $q_{1}$, the elbow pitch $q_{2}$, and the elbow roll $q_{3}$. Hence, the system admits a representation of the form \eqref{phmec} with 
$G=I_{3}$, $\mathcal{D} = \mathbf{0}$, and 
\begin{equation*}
 \begin{array}{rcl}
      M(q) &=& \begin{bmatrix} m_{1}(q_{2}) & 0 & \mathcal{I}_{3} \cos(q_{2})\\
    0 & \mathcal{I}_{2} + \mathcal{I}_{3} + m_{3} d_{c2}^{2} & 0\\
    \mathcal{I}_{3} \cos(q_{2}) & 0 & \mathcal{I}_{3}\end{bmatrix} \\[0.6cm]
    m_{1}(q_{2})&:=& \displaystyle\sum_{i=1}^{3}\mathcal{I}_{i} + m_{3} d_{c2}^{2} \sin^{2}(q_{2}),\\
    V(q) &=&  m_{3}g_{r}d_{c_{2}} \left(1  -  \cos(q_{2})\right) ,
 \end{array}
\end{equation*} 
where, $\mathcal{I}_{1}, \ \mathcal{I}_{2}$, and $\mathcal{I}_{3}$ denote the moments of inertia, $m_{3}$ and $d_{c2}$ are the mass and the distance to the center of mass of the second link, respectively, and $g_{r}$ represents the gravitational acceleration. The parameters of the system are provided in Table \ref{modelparam}.
\begin{table}[ht]
	\caption{Model parameters}\label{modelparam}
	\centering
	\renewcommand{\arraystretch}{1.2}
	\begin{tabular}{cc}
		\hline
		Parameter & Value \\ \hline
		$g_r$& 9.81       \\
		$d_{c2}$&  0.16     \\
		$m_3$&   1    \\
		$I_1$&   0.0054    \\
		$I_2$&   0.0768    \\
		$I_3$&   0.00211    \\ \hline
	\end{tabular}
\end{table}

In this system, only the shoulder roll is controlled directly by the rotation of a motor and the rest of the angles are controlled via differential drives. Therefore, the elbow pitch $q_2$ and the elbow roll $q_3$ are controlled by two motors, i.e., each angle is controlled by $u_{2}+u_{3}$ and $u_{2}-u_{3}$, respectively.

The control objective is to stabilize the system at a desired configuration $q_{*}$ while ensuring that 
\begin{equation*}
 \begin{array}{rcccl}
  \lvert u_{1} \rvert \leq 17.1007, &&
  \lvert u_2+u_3 \rvert \leq 7.901, &&
  \lvert u_2-u_3 \rvert \leq 7.901
 \end{array}
\end{equation*} 
to protect the motors. Additionally, since the robotic arm is not equipped with velocity sensors, we can measure only the positions.

\subsubsection{Implementation of the control law \eqref{umec}}\label{exp:1}

Following the results of Proposition \ref{pro:fully}, the saturated control law \eqref{umec} renders ${(q_{*},\mathbf{0},\mathbf{0})}$ globally asymptotically stable. To corroborate the effectiveness of the control approach, we perform an experiment under the initial conditions ${q(0) =(-2.257,-0.206,0.044)}$, where the PERA is stabilized at the desired configuration $q_{*} = (-1.81,\pi/2,0.78)$ with the control parameters given in Table \ref{controlparamns}. The results of the experiment are depicted in Figs. \ref{trajns} and \ref{inputns}. In the latter, we observe the saturation of $u_{1}$. On the other hand, in Fig. \ref{trajns}, we note steady-state errors in $q_{1}$ and $q_{3}$. These errors may be caused by several factors, such as the neglected damping, the asymmetry of the motors, or their dead zones. Hence, to remove these errors, we implement an integral-like term as it is explained in Section \ref{sec:rse}.

\begin{table}[ht!]
	\caption{Parameters for control law \eqref{umec}}\label{controlparamns}
	\centering
	\renewcommand{\arraystretch}{1.4}
	\begin{tabular}{cc}
		\hline
		Parameter & Value \\ \hline
		$\alpha$&   $(17,3,3.3)^\top$    \\
		$\beta$&     $(80,100,80)^\top$  \\
		$K_c$&   $\diag\{1,1,1\}$    \\
		$R_c$&   $\diag\{0.1,0.005,0.05\}$    \\ \hline
	\end{tabular}
\end{table}

\begin{figure}[ht!]
	\centering
	\includegraphics[width=\textwidth]{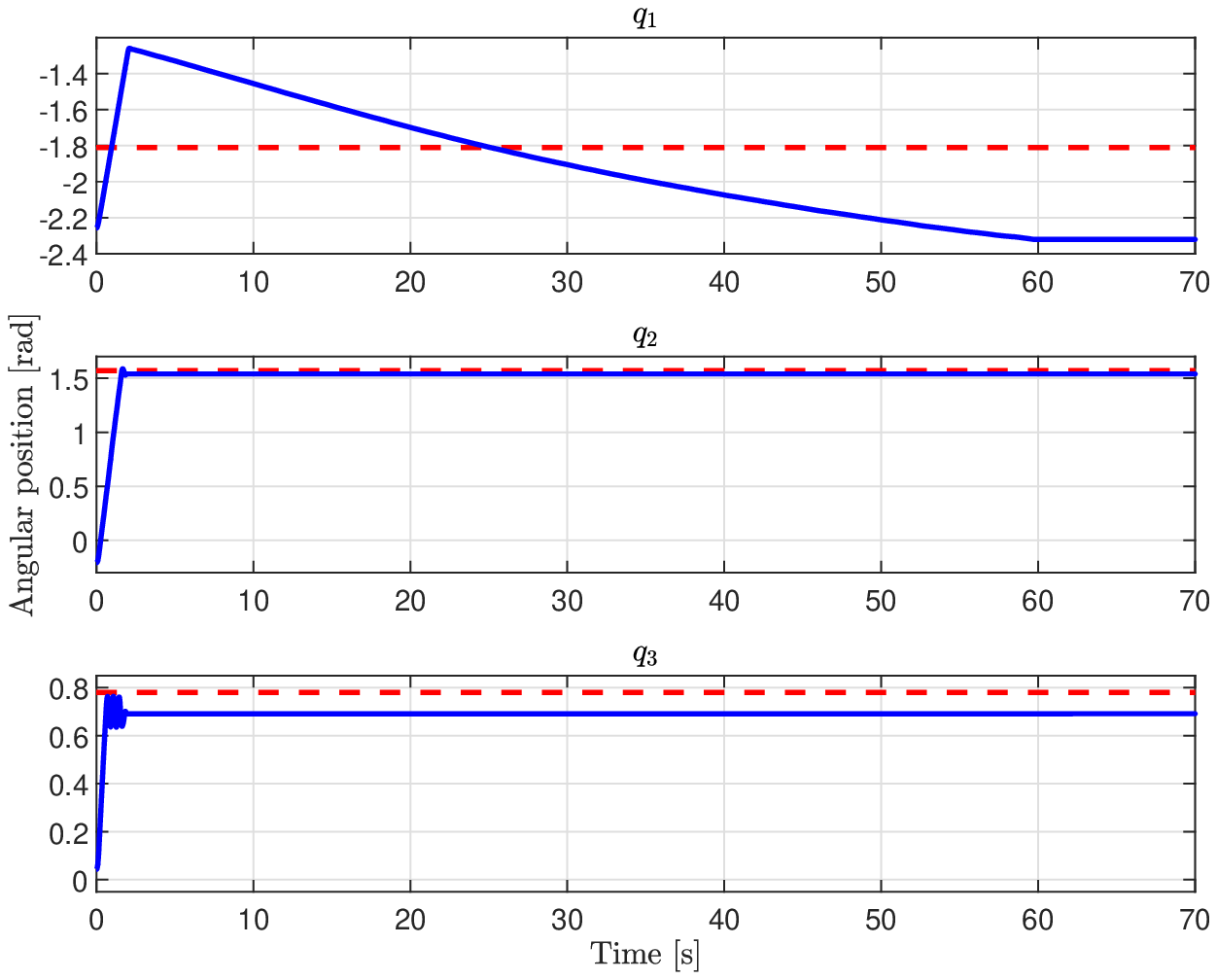}
	\caption{Angular positions (blue) and desired configurations (red).}\label{trajns}
\end{figure}
\begin{figure}[ht!]
	\centering
	\includegraphics[width=\textwidth]{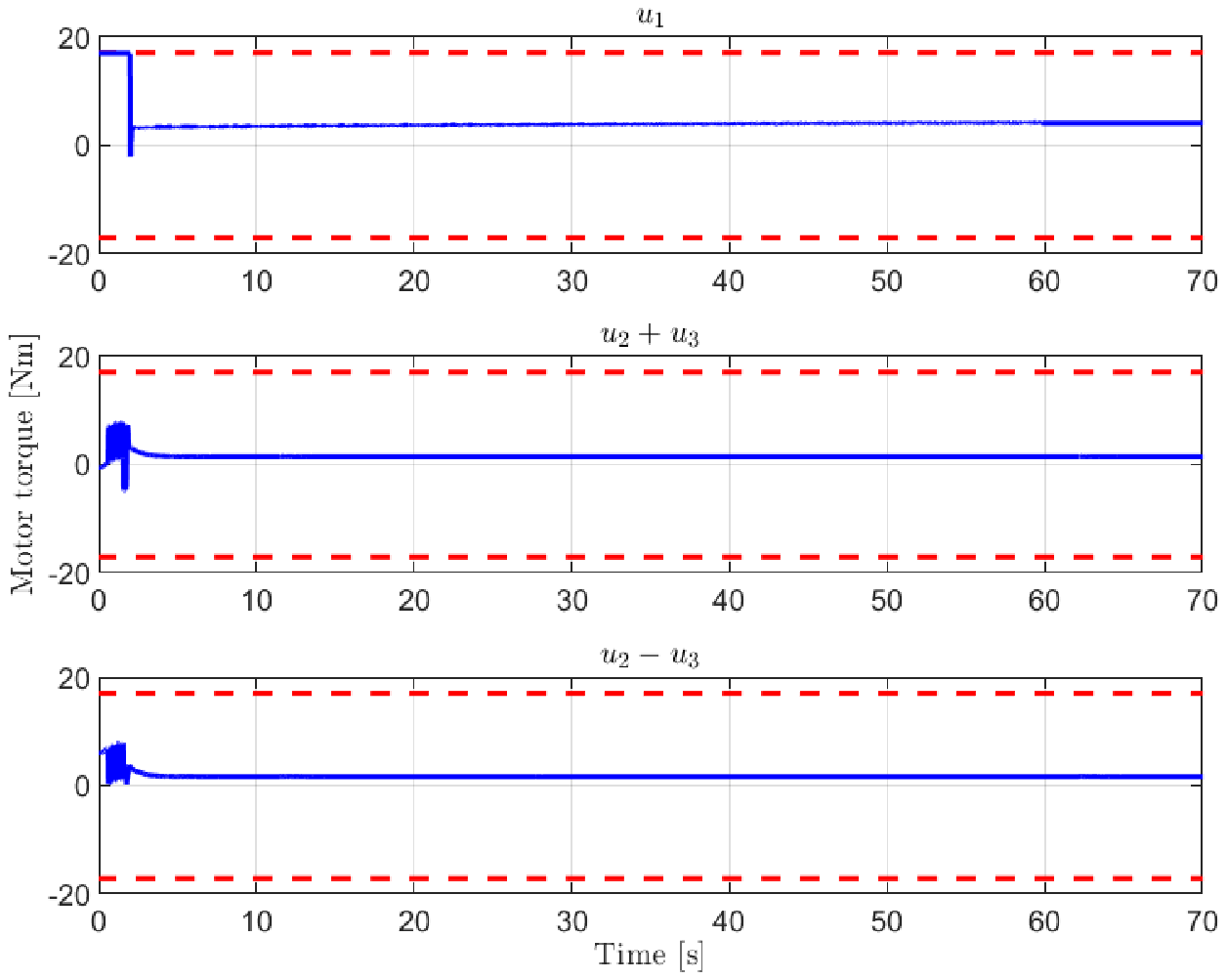}
	\caption{Torques provided by the motors (blue) and motors safety limits (red).}\label{inputns}
\end{figure}

\subsubsection{Implementation of the control law \eqref{conint}}\label{exp:2}

To remove the steady-state error from the results of Section \ref{exp:1}, we implement a filter of the form \eqref{dpsi} with 
\begin{equation*}
 \begin{array}{rcl}
  f_{\psi}(\psi) &=& -R_{\psi}\psi, \\ \Psi(\psi) &=& \displaystyle\sum_{i=1}^{3}e_{i}e_{i}^{\top}\alpha_{\psi_{i}}\beta_{\psi_{i}}\sech\left( \beta_{\psi_{i}}\psi_{i} \right), \\ u_{\psi_{i}}(\psi) &=& - \displaystyle\sum_{i=1}^{3}e_{i}\alpha_{\psi_{i}}\tanh\left( \beta_{\psi_{i}}\psi_{i} \right),
 \end{array}
\end{equation*} 
where $R_{\psi}$ is a diagonal matrix with positive entries. Accordingly,
\begin{equation}
 \left(\nabla f_{\zeta_{cl}}\right)_{*} = \begin{bmatrix}
     \mathbf{0} & M^{-1}_{*} & \mathbf{0} & \mathbf{0} \\
     -D_{c} & \mathbf{0} & -D_{c} & -D_{\psi} \\
     -R_{c}D_{c} & \mathbf{0} & -R_{c}K_{c}-R_{c}D_{c} & \mathbf{0} \\
     D_{\psi} & \mathbf{0} & \mathbf{0} & -R_{\psi}
                                                             \end{bmatrix}, \label{linPERA}
\end{equation} 
where
\begin{equation*}
 \begin{array}{rcl}
  D_{c}&:=&\diag\{\alpha_{c_{1}}\beta_{c_{1}},\alpha_{c_{2}}\beta_{c_{2}},\alpha_{c_{3}}\beta_{c_{3}} \} \\
  D_{\psi}&:=&\diag\{\alpha_{\psi_{1}}\beta_{\psi_{1}},\alpha_{\psi_{2}}\beta_{\psi_{2}},\alpha_{\psi_{3}}\beta_{\psi_{3}} \}.
 \end{array}
\end{equation*} 
Then, according to Proposition \ref{pro:int}, the augmented system has a globally asymptotically equilibrium at $(q_{*},\mathbf{0},\mathbf{0},\mathbf{0})$ if the control parameters of \eqref{conint} are selected such that all the eigenvalues of the matrix at the right-hand of \eqref{linPERA} have real part negative. To satisfy this condition, we propose the control parameters provided in Table \ref{controlparam}.

\begin{table}[ht!]
	\caption{Parameters for the control law \eqref{conint}}\label{controlparam}
	\centering
	\renewcommand{\arraystretch}{1.4}
	\begin{tabular}{cc}
		\hline
		Parameter & Value \\  \hline 
		$\alpha_c$&   $(6,1.4,1)^\top$    \\
		$\beta_c$&     $(120,120,120)^\top$  \\
		$\alpha_{\psi}$&     $(11,1.5,2.4)^\top$  \\
		$\beta_{\psi}$&  $(7,7,7)^\top$     \\
		$K_c$&   $\diag\{1,1,1\}$    \\
		$R_c$&   $\diag\{0.1,0.005,0.05\}$    \\
		$R_{\psi}$& $\diag\{1,1,35\}$			\\ \hline
	\end{tabular}
\end{table}

To corroborate that the steady-state errors are removed, we carry out experiments under initial conditions $q(0) =(-2.23,-0.212,0.086)$, considering the same desired configuration as in Section \ref{exp:1}. The results are shown in Figs. \ref{traj} and \ref{input}. We remark the absence of steady-state errors in the trajectories depicted in Fig. \ref{traj}, where the improvement with respect to the results of \ref{exp:1} is particularly notorious in $q_{1}$. Moreover, the saturation of $u_{1}$ is evident in Fig. \ref{input}. 
The video of this experiment can be watched at \url{https://www.youtube.com/watch?v=l-9DbTZvyD0}.

\begin{figure}[ht!]
	\centering
	\includegraphics[width=\textwidth]{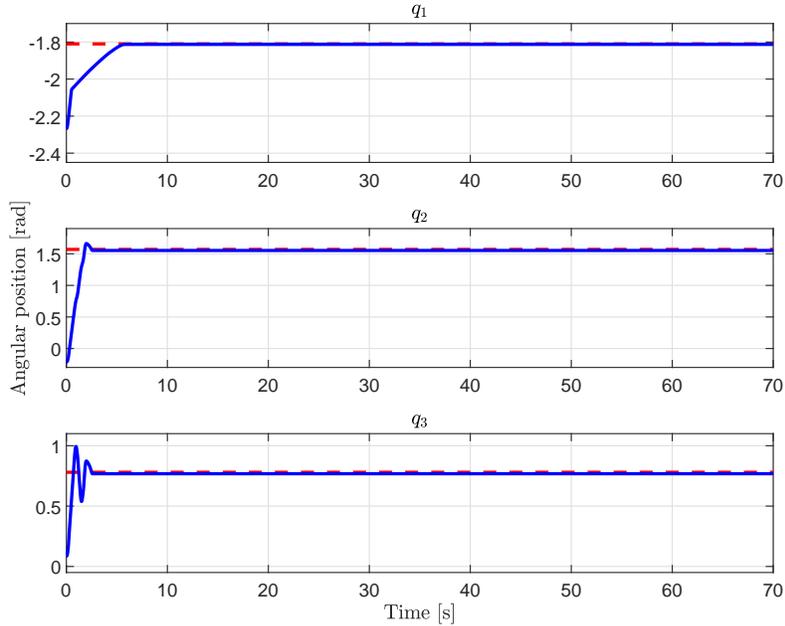}
	\caption{Angular positions (blue) and desired configurations (red).}\label{traj}
\end{figure}
\begin{figure}[ht!]
	\centering
	\includegraphics[width=\textwidth]{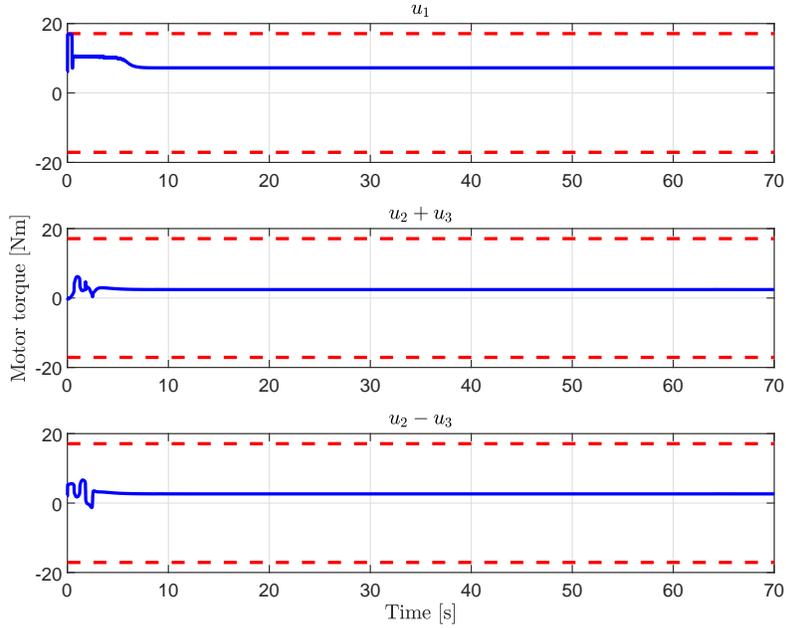}
	\caption{Torques provided by the motors (blue) and motors safety limits (red).}\label{input}
\end{figure}

\section{Concluding remarks and future work}\label{sec:7}
We have presented a PBC approach to design saturated controllers suitable for stabilizing a broad class of physical systems characterized by Assumptions \ref{ass1} and \ref{ass2}. Moreover, the proposed controllers do not require measuring the passive output to inject damping into the closed-loop system. Additionally, we have introduced a method to exploit the natural dissipation of the system to improve the performance of the controllers for systems with poor damping propagation. We have illustrated the applicability of the technique by controlling three systems in different physical domains, where the efectiveness of the methodology has been validated through simulations and experiments\\[0.2cm]
As future work, we aim to propose a constructive approach to tune the gains of the controllers to guarantee appropriate performance of the closed-loop system. 

\section*{Acknowledgements}

Pablo Borja and Jacquelien M.A. Scherpen thank Floris van den Bos for the fruitful discussions on the control of the PERA system.

\bibliography{bibliographyifac}
\bibliographystyle{plain}

\end{document}